\documentclass[11pt,english]{article}

\usepackage{amsfonts}
\usepackage{geometry}
\usepackage{mathrsfs}
\usepackage{enumerate}
\usepackage{color}
\usepackage[T1]{fontenc}
\usepackage[latin9]{inputenc}

\usepackage{amsmath,amssymb}
\usepackage{graphicx,caption}
\usepackage{tikz,ifthen}
\usepackage{subfig}
\usepackage{multirow}
\usepackage{nopageno}

\definecolor{myblue}{RGB}{80,80,160}
\definecolor{mygreen}{RGB}{80,160,80}
\usetikzlibrary{shapes,snakes,positioning,automata}
\usetikzlibrary{fit,shapes, chains, calc}
 \usetikzlibrary{arrows}
    \tikzstyle{vertex}=[circle,fill=black!25,minimum size=20pt,inner sep=0pt]
    \tikzstyle{edge} = [draw,thick,->]

\def\RR{\mathbb R}

\def\cA{\mathcal A}

\def\cS{\mathcal S}

\def\cM{\mathcal M}

\def\cW{\mathcal W}

\def\bx{\mathbf x}

\def\bx{\mathbf x}

\def\bq{\mathbf q}
\def\bu{\mathbf u}

\def\bw{\mathbf w}
\def\bzero{\mathbf 0}
\def\b1{\mathbf 1}

\newcommand{\qed}{\hfill$\square$\bigskip}
\newcommand{\raf}[1]{(\ref{#1})}

\newcommand{\argmax}{\operatorname{argmax}}

\newcommand{\PoA}{\operatorname{PoA}}

\newcommand{\hide}[1]{}

\newtheorem{theorem}{Theorem}
\newtheorem{lemma}{Lemma}

\newtheorem{example}{Example}

\newtheorem{definition}{Definition}
\newtheorem{claim}{Claim}
\newtheorem{observation}{Observation}

\title{Price of Anarchy in Algorithmic Matching of Romantic Partners}

\usepackage{authblk}

\author[a]{Andr\'es Abeliuk\thanks{A.A. and K.E. are joint first authors who contributed equally to this work.}}
\author[*b]{Khaled Elbassioni} 
\author[c]{Talal Rahwan}
\author[d]{\authorcr Manuel Cebrian}
\author[d,e]{Iyad Rahwan\thanks{To whom correspondence should be addressed. E-mail: irahwan@mit.edu}}

\affil[a]{Information Sciences Institute, University of Southern California, Marina del Rey CA, USA}
\affil[b]{Department of Computer Science, Khalifa University of Science \& Technology, Abu Dhabi, UAE}
\affil[c]{Computer Science, New York University, Abu Dhabi, UAE}
\affil[d]{Media Laboratory, Massachusetts Institute of Technology, Cambridge MA, USA}
\affil[e]{Institute for Data, Systems and Society, Massachusetts Institute Technology, Cambridge MA, USA}

\begin{document}
\date{}
\maketitle

\begin{abstract}
Algorithmic-matching sites offer users access to an unprecedented number of potential mates. However, they also pose a principal-agent problem with a potential moral hazard. The agent's interest is to maximize usage of the Web site, while the principal's interest is to find the best possible romantic partners. This creates a conflict of interest:  optimally matching users would lead to stable couples and fewer singles using the site, which is detrimental for the online dating industry.
Here, we borrow the notion of Price-of-Anarchy from game theory to quantify the decrease in social efficiency of online dating sites caused by the agent's self-interest.
We derive theoretical bounds on the price-of-anarchy, showing it can be bounded by a constant that does not depend on the number of users of the dating site. This suggests that as online dating sites grow, their potential benefits scale up without sacrificing social efficiency. Further, we performed experiments involving human subjects in a matching market, and compared the social welfare achieved by an optimal matching service against a self-interest matching algorithm.
We show that by introducing competition among dating sites, the selfish behavior of agents aligns with its users, and social efficiency increases. 
\end{abstract}


\maketitle


With online dating on the rise for the past decade, the pool of potential partners has increased from a few thousands of people that you could meet in your daily life, to millions of people. Space and time constraints to meet a partner through your existing social circles, are no longer a restriction with online dating sites; you can potentially be matched with any user in any place of the world. Online dating has enormous potential to ameliorate what is for many people a time-consuming and often frustrating activity. 
  ``Online dating is pervasive, and has fundamentally altered both the romantic acquaintance process and the process of compatibility matching''~\cite{finkel2012online}.
  Indeed, 11\% of American adults and 38\% of those who are currently single and looking for a partner have used online dating sites or mobile dating apps~\cite{smith2013online}. In short, online dating is changing the way we find romantic partners~\cite{rosenfeld2012searching}.

The vast amount of people using dating services imposes a new problem to service providers. The set of potential matches are too large for a single user to search exhaustively. As such, the system must present or suggest a smaller, more tractable subset of potential matches to the user.
This entails that online dating sites must use an algorithm to curate what is presented to users, based on the large amounts of data the system collects from them. For example, {\em OkCupid} claims\footnote{https://theblog.okcupid.com} that their algorithm matches people using ``match percentages'', which basically quantifies how much users have in common. However, none of these algorithms are fully disclosed.

In the early 70's, Herbert Simon presaged that wealth of information would create attention scarcity, and hence, would become a new type of currency \cite{simon1971designing}. Indeed, we see technology companies competing to win people's attention, clicks and loyalty.  This notion, coined as ``attention economy'' \cite{goldhaber1997attention,falkinger2007attention}, is more relevant than ever, and is considered one of the most important determinant of business success. For example, in e-commerce, often more emphasis is being placed on users' future visits and purchases than on maximizing their current satisfaction \cite{smith2002loyalty}. 
Economic theory of ``\textit{planned obsolescence}'' predicts that oligopolists will profit from producing goods with shorter durability, to incentivize the consumers to purchase more frequently \cite{bulow1986economic}. In fact, there is evidence of planned obsolescence in the textbook market \cite{iizuka2007empirical}, where publishers introduce new editions to decrease the value of old ones. This raises the following question:

\emph{How would the social efficiency of online dating website be affected by planned obsolescence, whereby matches with shorter expected ``durability'' are suggested deliberately, and better matches are subsequently introduced to decrease the ``value'' of old matches, all in the hope of maximizing the likelihood of users coming back to the Web site?}


Put differently, since the business model of most online service providers relies heavily on the number of users they have, this creates a fundamental conflict of interest: {\em optimally matching users would lead to stable couples and fewer singles using the site, which is detrimental for the business.}  Hence, a self-interested dating service provider that maximizes its revenue will have as one of its goals the maximization of user \textit{engagement}, which is seldom aligned with maximizing user \textit{utility}. As such, while dating sites offer users access to an unprecedented number of potential mates, they also pose a principal-agent problem with a potential moral hazard (e.g., \cite{balmaceda2010cost}) as the agent (online dating site) has access to more information than the principal (the user).

Driven by these observations, this paper aims at quantifying the social welfare loss that users may experience from a dating site that maximizes engagement instead of user utility. In so doing, we take a first step toward quantifying the maximum social-welfare loss that could result from online dating sites acting selfishly. In particular, we model the centralized matching as a classical weighted matching on bipartite graphs~\cite{papadimitriou1982combinatorial}---with the novel difference being that self-interest behavior is represented by a different objective function to the classic objective of maximizing the sum of weights. Namely, the goal of dating sites is to maximize the users that remain in the system, according to a Markov decision process. To our knowledge, this is the first attempt to model the detrimental effects of selfish (centralized) matchings.

\section{Algorithmic Matching}
Matching in dating sites is an inherently online problem since there is no complete information about the arriving users to the system in advance, but obtained incrementally.
In contrast, we first consider an offline version of the matching problem, in which all users are known in advance and are in the system from the beginning. The online matching problem is analyzed in the Supplementary Information (Section S3).

From the perspective of a social designer, the problem of finding partners for everyone as to maximize social welfare is modeled as a weighted matching on bipartite graphs, where weights represent how well-matched two individuals are. Here, we assume that men's and women's preferences are known. 

Formally, let $\cM:=[m]$ and $\cW:=[n]$ be two finite sets of $m$ ``men'' and $n$ ``women'', respectively. Each individual $i$ is characterized by a vector: $\chi(i)\in\RR_+^d$ if $i\in \cM$ and $\kappa(i)\in\RR_+^d$ if $i\in \cW$. This vector representation is very general, and characterizes each individual through a set of relevant features, such as age, personality trait, film preferences, etc. Let $w_{ij}:=f(\chi(i),\kappa(j))\in[0,1]$ be a measure of how well-matched individuals $i$ and $j$ are.

For the offline model, assume all elements in $\cW$ and $\cM$  are  known by the agent (or system).  The system decides the (partial) distribution $\bx_i:=(x_{ij}:~i\in\cM,~j\in\cW,~\sum_{j\in\cW}x_{ij}\le 1)$ for all $i \in \cM$, and assigns an element $j\in \cW$ to $i$ with probability $x_{ij}$. 
Define the {\it level of satisfaction} or {\it utility} of $i\in\cM$ to be $u_i:=\sum_{j\in\cW} w_{ij}x_{ij}\in[0,1]$.


A {\it fair} agent (from the users' point of view) would assign $j$ to $i$ so as to maximize 
$f(\bu):=\sum_{i}u_i.$ That is, it attempts to solve the following optimization problem:
\begin{align}
z^*_f(\bw):=&\max_{\bx}\textstyle\sum_{i}u_i \label{o-f}\tag{$O_f(\bw)$}\\
\text{s.t.}& \qquad u_i:=\textstyle\sum_j w_{ij}x_{ij},~~\text{ for all }i\in\cM\label{e1}\\
&\qquad\textstyle\sum_{j\in\cW}x_{ij}\le 1,~~\text{ for all }i\in \cM,\label{e2}\\
&\qquad\textstyle\sum_{i\in\cM}x_{ij}\le 1,~~\text{ for all }j\in \cW,\label{e3}\\
&\qquad x_{ij}\ge 0,~~\text{ for all }i\in \cM\text{ and }j\in \cW.\label{e4}
\end{align}
Given that the problem is modeled on bipartite graphs, this linear program will yield integer solutions 
\cite{wolsey1993integer}. In other words, there always exists an optimal integer solution to the fair agent matching problem where every element $i \in \cW$ is matched only with one element $j \in \cM$ or is not match at all.


\subsection{Modeling Self-interested Behavior} \label{sec:markov} 
Next, we introduce the concept of ``user engagement'' into the matching problem. The objective function of a self-interested designer of a matching service is to maximize this value. 

We model the situation for user $i$ in relation with the agent (i.e., matching service) by a {\it Markov decision process} with two states $s_{i,1}$ and $s_{i,2}$, representing the states of being in and out of the system,  respectively. The Markov chain is depicted in Figure~\ref{fig:Markov1}. At state $s_{i,1}$, there is an infinite action set representing the decision $u_i\in[0,1]$. For each such action, the user moves to state $s_{i,2}$ with probability one, where he or she stays there with probability $1-q_i(u_i)$ (for either a successful relationship or a complete dissatisfaction with the matching service), or comes back to the system (state $s_{i,1}$) with probability $q_i(u_i)$, after having received a utility $u_i$.
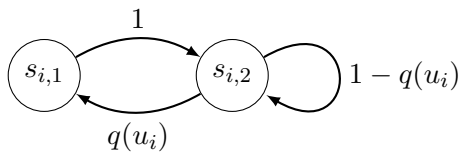
\begin{figure}[h]
\centering
\begin{tikzpicture}[node distance=2.5cm,->,>=latex,auto,
  every edge/.append style={thick}]
  \node[state] (1) {$s_{i,1}$};
  \node[state] (2) [right of=1] {$s_{i,2}$};  
  \path (1) 
            edge[bend left]  node{$1$}   (2)
        (2) edge[in=-30,out=30, loop] node{$1-q(u_i)$}  (2)
            edge[bend left] node{$q(u_i)$}     (1);
\end{tikzpicture}
\caption{Markov process with two states $s_{i,1}$ and $s_{i,2}$, representing the states of being in and out of the system, for user $i$ having received utility $u_i$.} \label{fig:Markov1}
\end{figure}

Let $q_i(u_i)$ be the probability that individual $i$, having received the utility $u_i$, will come back to the system. We will make the natural assumptions\footnote{Some of these assumptions can be relaxed, but we will not do this here to keep the presentation simple.} that $q_i(u_i)$ satisfies:
\begin{itemize}
\item[(A1)] $q_i(0)=q_i(1)=0$;
\item[(A2)] $q_i(u_i)$ is twice-differentiable in $u_i\in[0,1]$;
\item[(A3)] $q_i(u_i)$ is strictly concave in $u_i\in[0,1]$.
\end{itemize}
The intuition behind the assumptions is that users who are dissatisfied by the service (i.e., receive a low utility) are less likely to return and in the extreme case of receiving zero utility, the user will never return to the system. On the other hand, users who are paired with good matches have less incentives to look for a new partner and in the extreme case of having received the maximum utility of one  (i.e., a  perfect match), the user will never return to the system.

Let $\bq:=(q_1,\ldots,q_m)$ be the vector of probability functions.
For $k\in\{1,2\}$, let $\pi_{i,k}(u_i)$ be the limiting probability that the user will be in state $s_{i,k}$, given that the decision made by the agent for the user is $u_i$. Then we have:
\begin{equation}\label{eq:pi}
\pi_{i,1}(u_i)=\frac{q_i(u_i)}{1+q_i(u_i)}\text{ and } \pi_{i,2}(u_i)=\frac{1}{1+q_i(u_i)}.
\end{equation}
 
\begin{observation}\label{obs:1}
$\pi_{i,1}(u_i)$ satisfies (A1), (A2) and (A3) (with $q_i$ replaced by $\pi_{i,1}$).
\label{ob1}
\end{observation}


\medskip

The agent would {\it selfishly} allocate users so as to maximize the {\it expected number of users returning to the system}
$s(\bu;\bq):=\sum_{i}\pi_{i,1}(u_i)$. That is, the selfish central designer attempts to solve the following optimization problem:
\begin{align}
z^*_s(\bw;\bq):=&\max_{\bx}\sum_{i}\pi_{i,1}(u_i)\label{o-s}\tag{$O_s(\bw;\bq)$} \\
\text{s.t.}& \qquad \text{ Constraints \raf{e1}-\raf{e4}}.
\end{align}

In the {\it offline} case, both problems \raf{o-f} and \raf{o-s} can be solved optimally in polynomial time via linear programming and convex programming, respectively. Next, we quantify the loss of efficiency in the system induced by the selfish behavior of the central designer with the notion of Price of Anarchy (PoA) \cite{koutsoupias1999worst}.

\begin{definition}[Price of Anarchy]\label{PoA}
Let $\bu^*_f$ and $\bu_s^*(\bq)$ be the optimal solutions of the problem for the fair agent \raf{o-f} and selfish agent \raf{o-s}, respectively.
Define the {\it price of anarchy} $\PoA=\PoA(\bq)\in[0,1]$ as the ratio between the worst possible selfish matching and the social optimum. Formally,
$$
\PoA(\bq):=\min_{\bw\ge \bzero}\frac{f(\bu_s^*(\bq))}{f(\bu^*_f)}.
$$
%
%
\end{definition}

\section{Bounds on the Price of Anarchy}
In this section, we present our main theoretical result stating that the price of anarchy can be bounded by a constant that depends only on the functions $q_i$. This result entails two conclusions: (1) as online dating sites grow, their potential benefits scale up without sacrificing social efficiency; (2) the loss of social utility is solely driven by user's behavior, which is exogenous to the matching service, and thus, if modified can improve the efficiency of the system. In other words, the system can only exploit users to the extent that they let it. This idea will be further explored in the next section by modelling users' behavior in a competitive market instead of in a monopoly.

Formally, we show that for any vector of probability functions $\bq=(q_1,\ldots,q_m)$ satisfying (A1)-(A3), the price of anarchy ($\PoA(\bq)$) can be bounded by a constant that depends only on the functions $q_i$, but not on the number of users $m$.
\begin{theorem}\label{t1}
Let $H(\bq):=\max_{i\in\cM}\{q_i'(0)\}>0$, and $c>0$ be the unique solution of the equation 
\begin{equation}\label{eq:c}
c=\frac{H(\bq)}{2}\cdot L(\bq,c),
\end{equation}
where $L(\bq,c)=\min_{i\in\cM}\overline{u}_i(c)$, and where for $i\in\cM$,  $\overline{u}_i(c)\in[0,1]$ denotes the unique positive solution of the equation 
\begin{equation}\label{eq:u}
c=\pi_{i,1}'(u_i) = \frac{q_i'(u_i)}{(1+q_i(u_i))^2}.
\end{equation}
Then, 
$$\PoA(\bq)\ge L(\bq,c)/2.$$
\end{theorem}
\begin{figure}[t!]
\begin{centering}
\includegraphics[width=0.9\columnwidth]{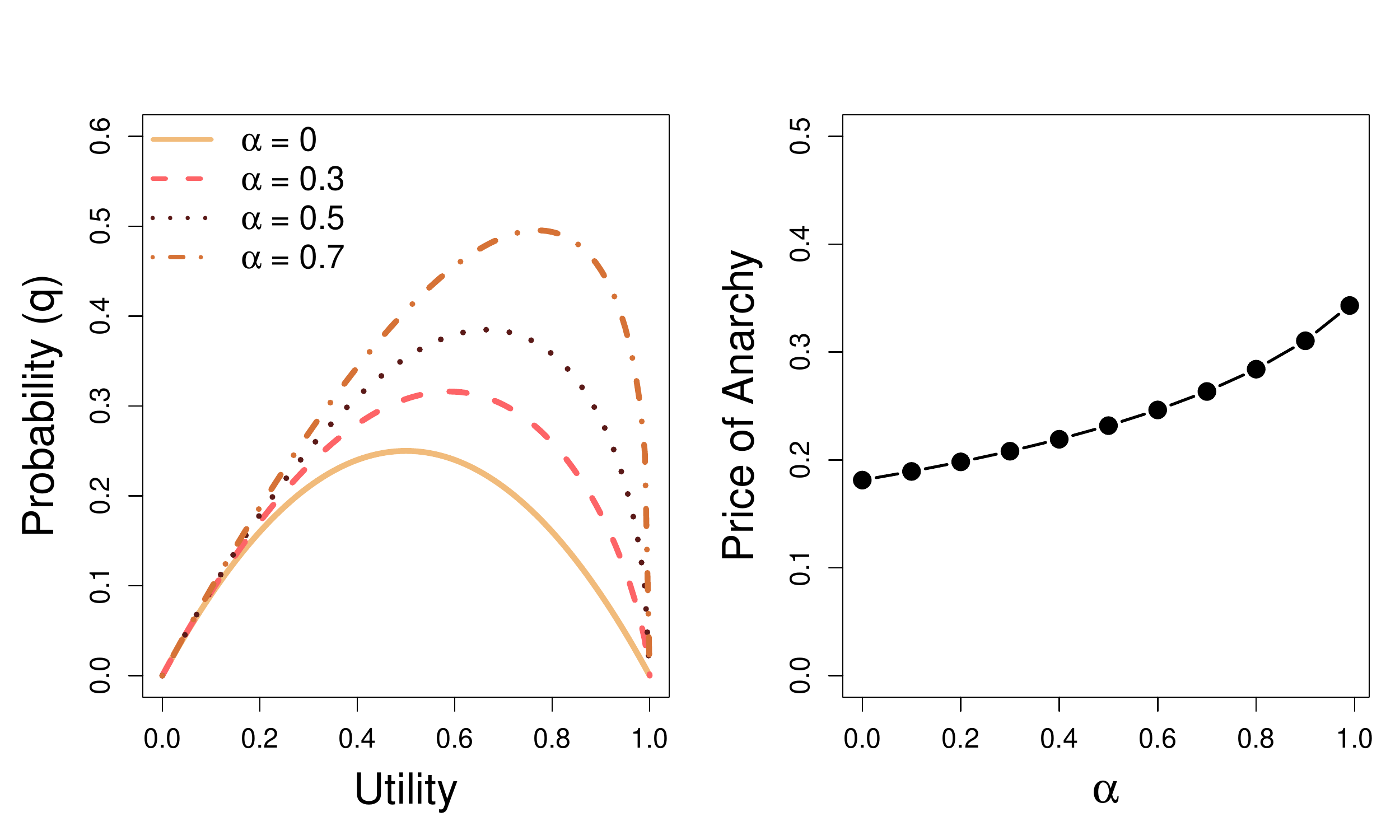}
\par\end{centering}
\caption{\label{fig:example1} Left: The figure shows probability functions $q_i(u):=u(1-u)^{(1-\alpha)}$ for different $\alpha$'s. Right: The Price of Anarchy as a function of $\alpha$.  }
\end{figure}
\begin{example}\label{ex1}
Let us suppose that $q_i(u_i):=u_i(1-u_i)^{(1-\alpha)}$ for all $i=1,\ldots,m$, which satisfies assumptions (A1)-(A3) for $0\leq \alpha<1$. Figure~\ref{fig:example1} depicts the probabilities for different $\alpha$'s.  
 Next, for ease of exposition, we find the Price of Anarchy  for $\alpha=0$ and show in Figure~\ref{fig:example1} the results for all $\alpha$'s.  
Let $\overline{u}_i=u$ be the unique solution of Equation \raf{eq:u}: $c=(1 - 2 u)/((1 - u) u + 1)^2 $. Also, we have that $H(\bq)=q'(0)=1$, and $c$ is chosen to satisfy \raf{eq:c}, which means that $c=u/2$. The unique solution in $(0,1)$ is $u=0.363$, which implies that $\PoA(q)\ge0.1815$.
\end{example}

\subsection{The Invisible-Hand Effect}\label{sec:competition}
So far, we have considered the monopolistic case where users only have one option to go and get matched. However, in a competitive market, utility-maximizing users who freely choose among competing dating sites will steer the system  towards a better social utility. Next, we explore the effect of competition among dating sites. To this end, we assume that users' probability to remain dating in a particular site is proportional to the utility they have previously received by that site. That is, the conditional probability of a user returning to the same dating site having received previously a utility of $u_i$, given that he or she wants to date again is $\Pr(\text{return}|\text{dates again})=u_i$. The probability of dating again is define as in the previous section, i.e., $\Pr(\text{dates again})=q(u_i)$. Thus, $\Pr(\text{return}) =\Pr(\text{return}|\text{dates again})\Pr(\text{dates again})=u_i \cdot q(u_i)$. Similarly, given that a user wants to date again, he or she will consider a competing dating site with probability $1-u_i$. 


\begin{figure}[h]
\centering
\begin{tikzpicture}[node distance=3.5cm,->,>=latex,auto, scale=1,
every node/.style={scale=1.},
  every edge/.append style={thick}]
  \node[state] (1) {$s_{i,1}$};
  \node[state] (2) [right of=1] {$s_{i,2}$};  
  \node[state] (3) [right of=2] {$s_{i,3}$}; 
  \path (1) 
            edge[bend left]  node[below]{$1$}   (2)
        (2) edge[loop below] node{$1-q(u_i)$}  (2)
            edge[bend left] node{$q(u_i)\cdot u_i$} (1)
            edge[bend right] node[below] {$q(u_i)(1-u_i)$}  (3)
        (3) edge[loop above]  node{$1-\epsilon$} (3)
            edge[bend right=37]  node[above]{$\epsilon$}   (1);
\end{tikzpicture}
\caption{Markov process with three-states $s_{i,1}$, $s_{i,2}$ and $s_{i,3}$, representing, respectively, the states of being dating in the system; not looking for a partner; and being out of the system (or in the competition), for user $i$ having received utility $u_i$.} \label{fig:Markov2}
\end{figure}
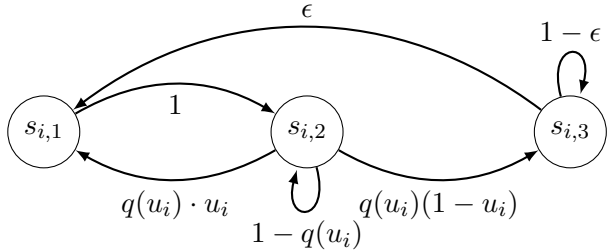

We extend the Markov decision process defined in Section~\ref{sec:markov}, to model the situation with competition for user $i$ as a three-state Markov process. The first state represents looking for a partner with the dating agent; the second represents being in the system but not looking for a partner; the third represents leaving to a competitor dating site. We assume that once a user leaves, there is a small probability $\epsilon$ of returning back. $\epsilon>0$ ensures the Markov chain is ergodic and has a stationary distribution.
The Markov chain is depicted in Figure~\ref{fig:Markov2}.

From state $s_{i,1}$, the user moves to state $s_{i,2}$ with probability $1$, where he or she stays there with probability $1-q_i(u_i)$ (for either a successful relationship or a complete dissatisfaction with the agent), or comes back to the system (state $s_{i,1}$) with probability $q_i(u_i) \cdot u_i$ (notice that without competition this probability was $q_i(u_i)$). The user can also move to state $s_{i,3}$ (for a user switching to the competition) from state $s_{i,2}$ with probability $q_i(u_i)(1-u_i)$, where he or she stays there with probability $1-\epsilon$ ($\epsilon$ being smaller as the competition is higher), or comes back to the system with probability $\epsilon$. Then we have
\begin{equation}\label{eq:competition}
\pi_{i,1}(u_i) = \frac{q(u_i)}{1+q_i(u_i)+\frac{q_i(u_i)}{\epsilon}(1-u_i)}.
\end{equation}
Note that $\pi_{i,1}(u_i)$ is not necessarily a concave function, and thus, the result from Theorem~\ref{t1} cannot be directly applied. However, 
we show in the following lemma that as competition between service providers grows, the self-interested behavior of system designers converges to maximize the social utility of users. We present empirical validation in Example~\ref{ex:2}.
\begin{lemma}\label{lem:1}
$\argmax_{u_i}\pi_{i,1}(u_i)=1$,
as $\epsilon \rightarrow 0$.   
\end{lemma}



\section{Experimental Study}
Our theoretical results provide bounds on the Price of Anarchy, which is defined as the worst-case efficiency level. We showed that even in a monopolistic setting, the price of anarchy can be bounded by a constant that does not depend on the number of users, but only depends on users' implicit expectation of the system. However, our theoretical results are constrained to the functional form $q_i$ dictating users behavior, which in practice can vary from one user to another and even change over time (although, we do expect that in the average, the overall behavior follows assumptions (A1) to (A3) from Section~\ref{sec:markov}). To overcome this issue, we designed a human subject experiment to test our results, where $q$ is learned from the participants.  

In a series of experiments, we compared the social welfare achieved by the optimal matching algorithm against the selfish matching algorithm. 
We operationalized the algorithmic matching market as a centralized multi-armed bandit problem. Figure~\ref{fig:infographic} provides a high-level illustration of the experiment. In more detail, participants arrive into the market and are assigned to a slot machine that gives them a stochastic payoff. Participants play for ten rounds, at each round deciding between three options: (1) continue with the slot machine assigned; (2) request from the central designer to be assigned a new slot machine; (3) or take the risk-free outside option and finish the game. To make this analogy more concrete, payoffs capture going out on a date with someone and observing the value of their match; option (1) would represent having multiple dates with the same match and engaging in a long-term relationship; option (2) represents going back to the dating site to find a new match; and finally option (3) would be users that either move to a different dating site or stop using online dating altogether.


\begin{figure}
\begin{centering}
\includegraphics[width=0.5\columnwidth]{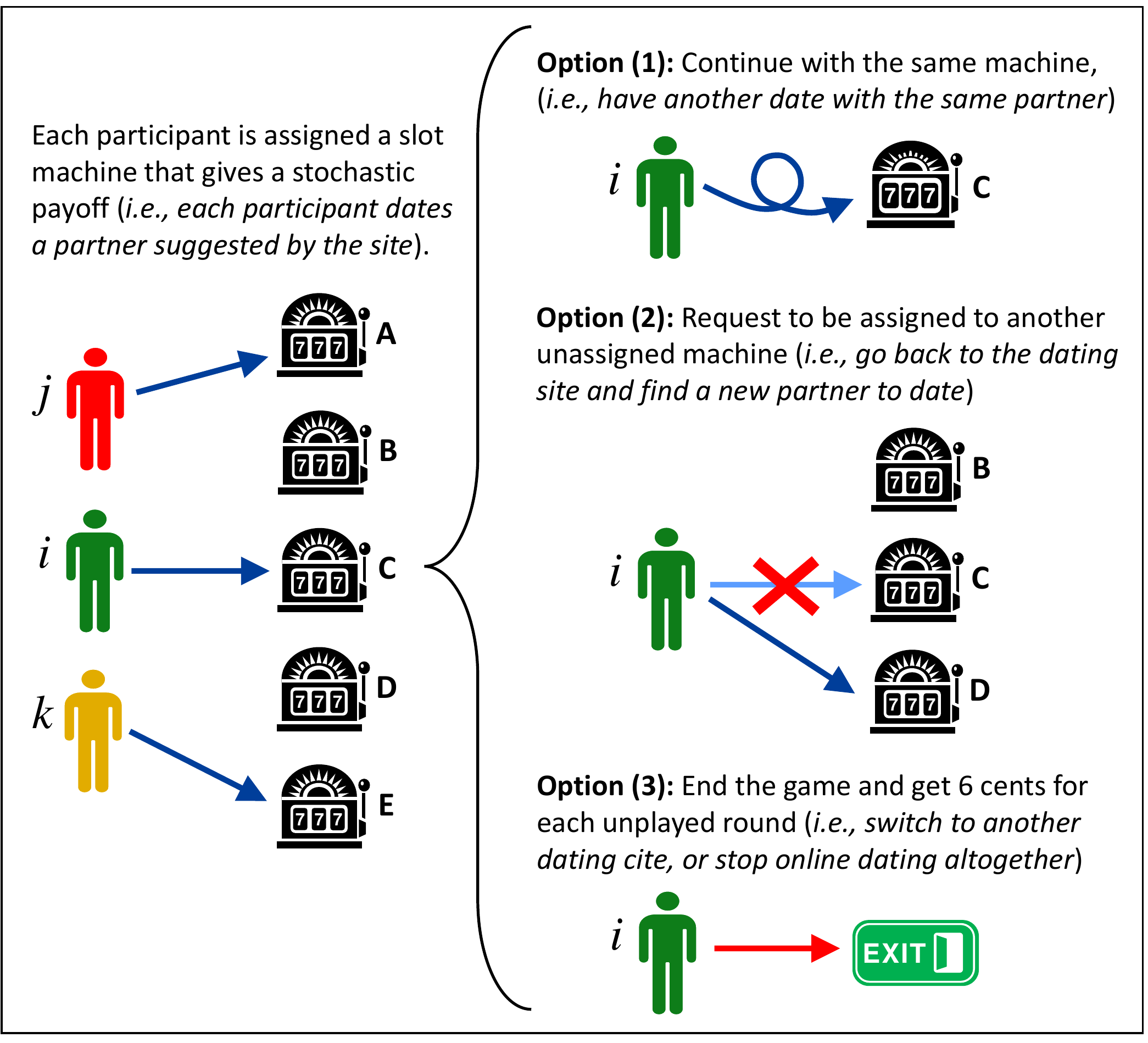}
\par\end{centering}
\caption{\label{fig:infographic} An illustration of the experimental study, showing the initial assignments of participants, and the three options available in each of the 10 rounds of the experiment.}
\end{figure}

\begin{figure*}
\begin{centering}
\subfloat[\label{fig:Study_A_dist}Study~A]
{\begin{centering}
\includegraphics[width=0.5\linewidth]{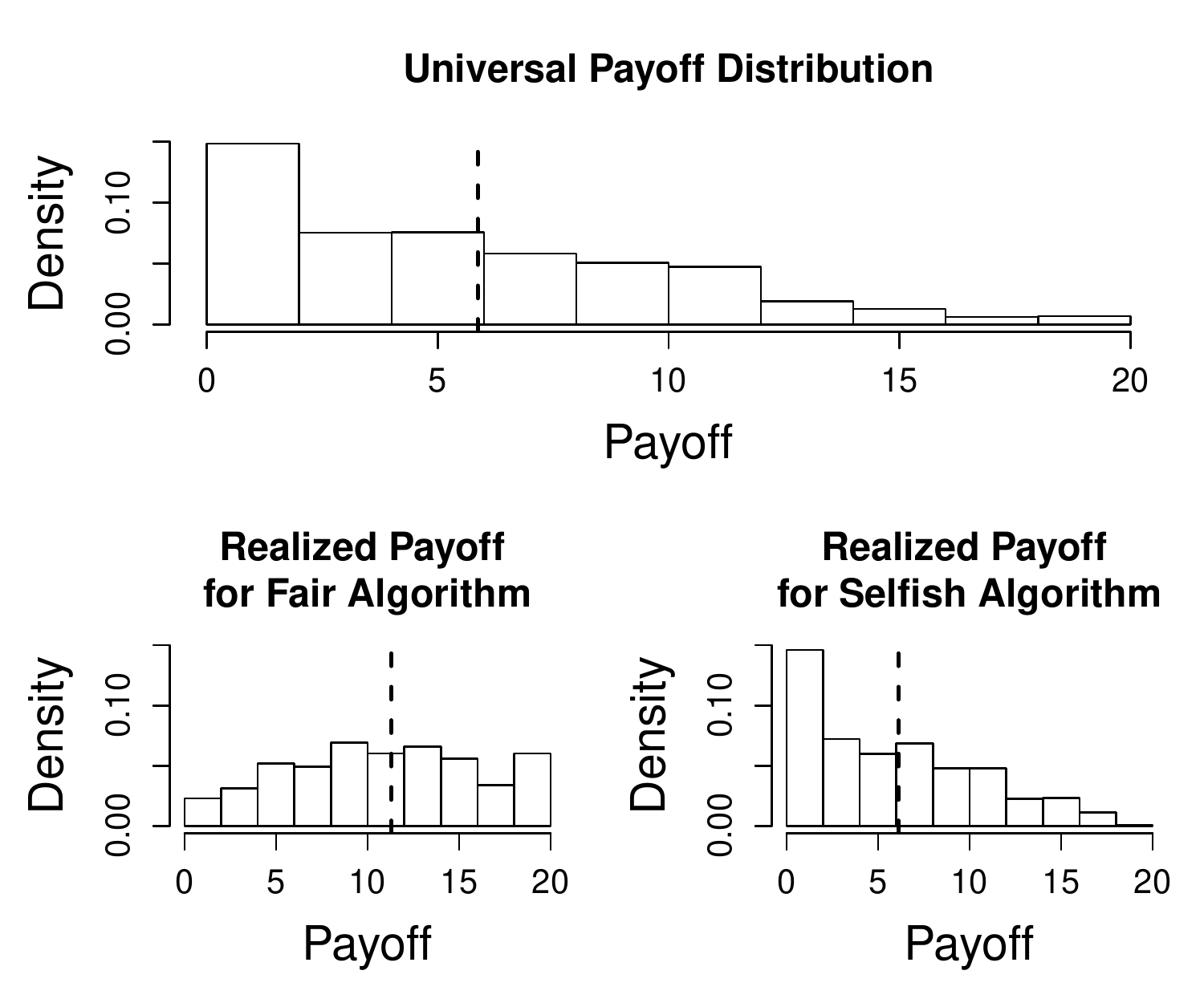}
\par\end{centering} 
}
\subfloat[\label{fig:Study_B_dist}Study~B]
{\begin{centering}
\includegraphics[width=0.5\linewidth]{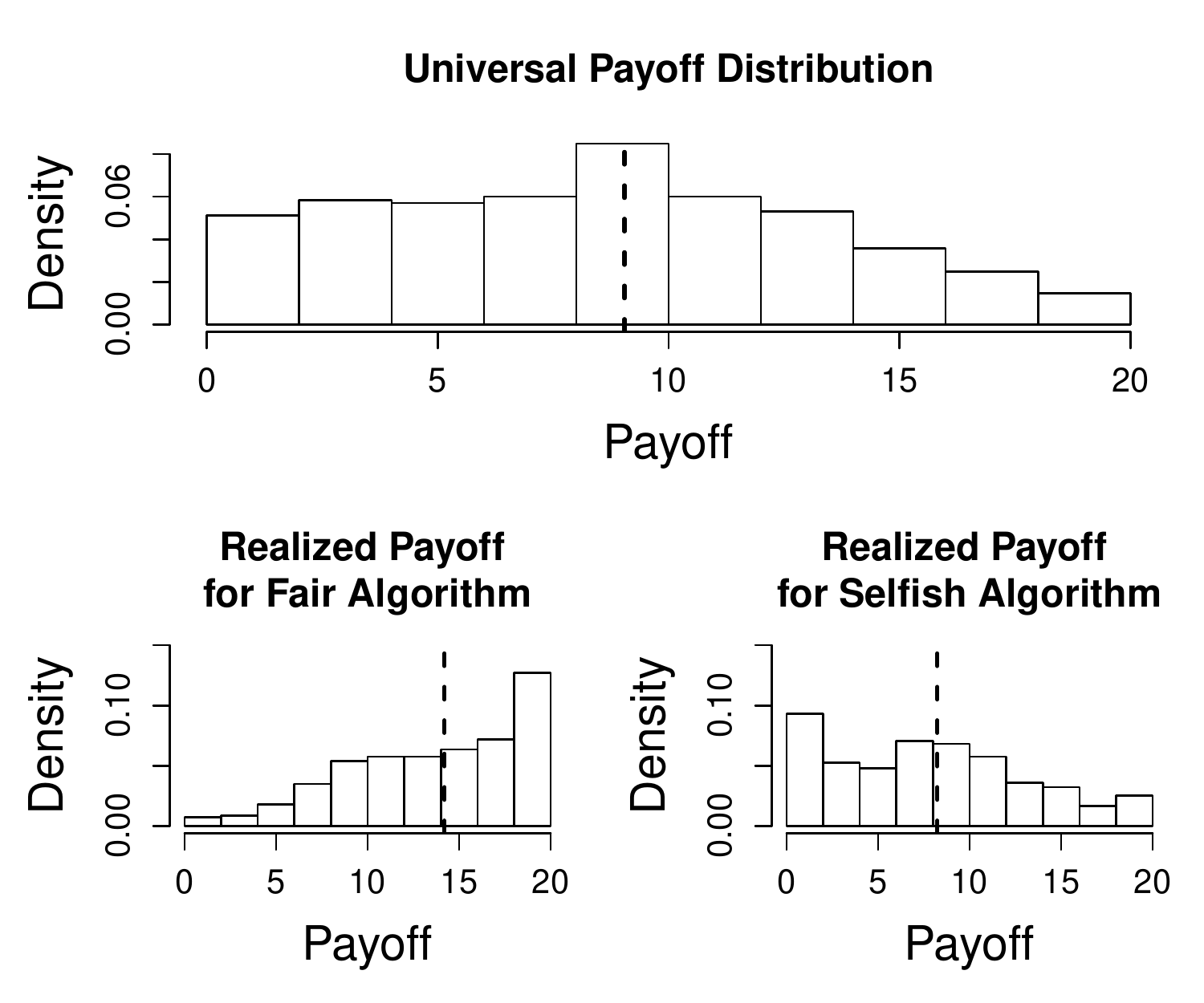}
\par\end{centering}
}

\subfloat[\label{fig:Study_C_dist}Study~C]
{\begin{centering}
\includegraphics[width=0.5\linewidth]{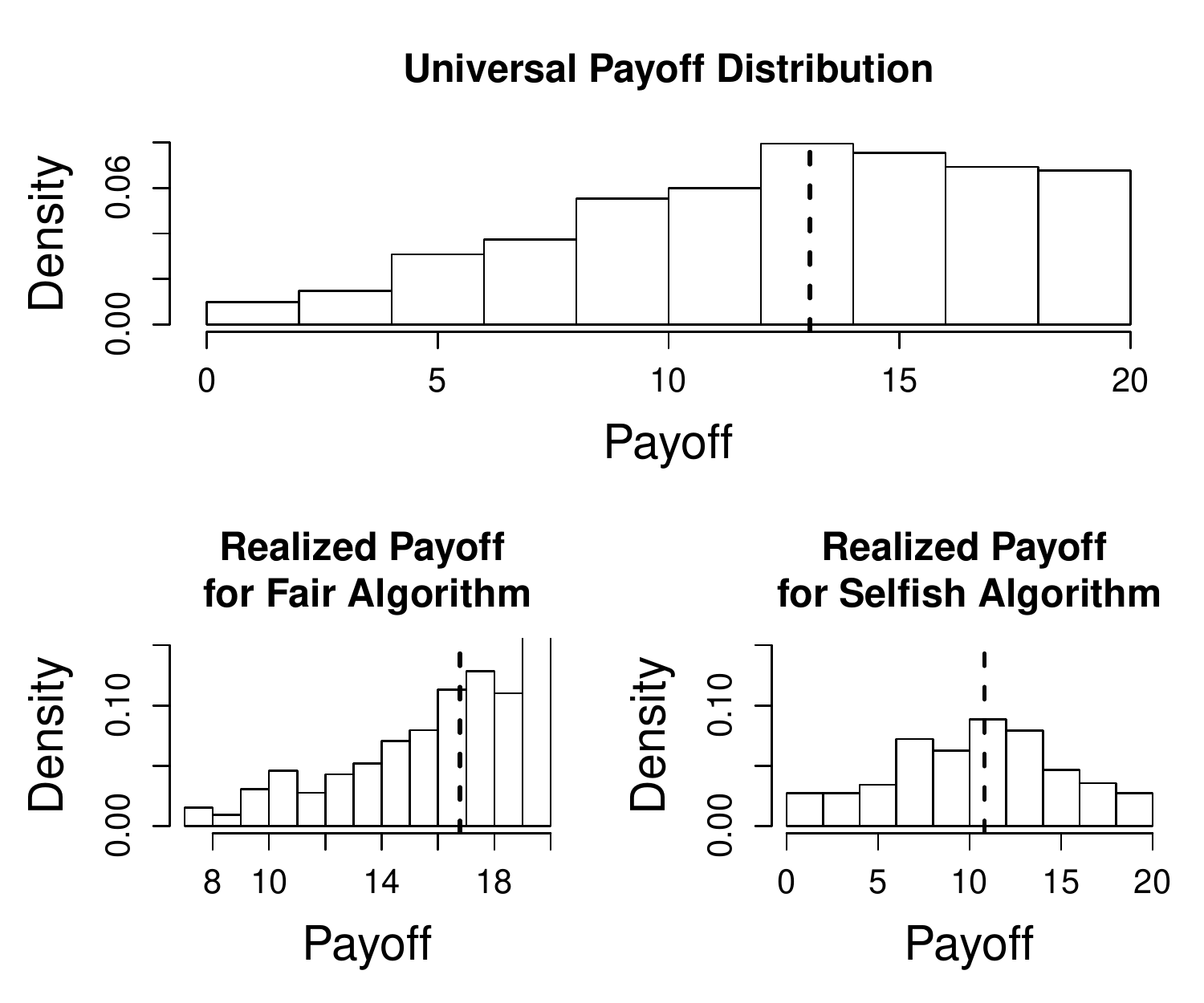}
\par\end{centering}
}
\par\end{centering}
\caption{\label{fig:Distribution-of-payoffs}Distribution of payoffs. The top distributions of each sub-figure correspond to the payoff distribution of all user-slot pairs for each study. 
The bottom distributions depict the realized payoffs experienced by participants. That is, it correspond to the payoff distribution of user-slot pairs that were matched. Vertical dashed lines represent the mean of the distribution. The universal distributions correspond to the realized payoffs of a matching algorithm that randomly assigns users to slots. Notice that for the Selfish algorithm there is a shift towards the left compared to the universal (slightly decreasing the mean), whereas for the Fair algorithm the shift is towards the right and a significant increase in the mean.}
\end{figure*}

The algorithm used to match players takes two forms. The Fair matching algorithm assigns agents to slot machines in order to maximize the social utility of all agents.  The Selfish matching algorithm assigns agents to slot machines in order to maximize the number of agents requesting re-assignations to the system. One of the advantages of this experimental design is that the central designer has perfect information on the expected payoffs, which allows us to quantify the impact of the matching algorithms without having to deal with the challenges of inferring and predicting the preferences of users, which is out of the scope of this paper. 
The interested reader is referred to~\cite{hitsch2010makes} for a work on
estimating mate preferences from an online dating service. 

Therefore, the Fair matching algorithm is given by solving problem $O_f(\bw)$. However, for the selfish matching, in order to solve $O_s(\bw;\bq)$ we need to learn the probability $q_i(u)$ that a user having received utility $u$, will come back to the system on round $i$. Given that individual data is sparse, we assume users are homogeneous and start with a prior probability $q_1(u)$ satisfying assumptions (A1) to (A3). Then, we update the probability based on the past choices of participants as follows:
\[
q_{i+1}(u) = \alpha q_i(u) + (1-\alpha)f_i(u),
\]
where $f_i(u)$ is the fraction of participants who, having received a payoff $u$ at round $i$, requested to be re-assigned; see Figures S5 to S7 for the distribution of $f_i(u)$. 

Each instance of the game consists of $2n$ participants which are evenly randomized between the two experimental conditions: the Fair matching algorithm and the Selfish matching algorithm.
Each player $i\in [n]$ has a mean payoff $\mathbb{E}[p_{ij}]= w_j+\epsilon_i$ for each of the $j\in [m]$ slots machines, where $\epsilon_i$ is a realization of a normally distributed variable with zero mean and variance of $3$. Thus, rewards from machines are correlated across participants, meaning that there is competition for high yielding slot machines. The game is played in real-time (or synchronously) with other players, such that any slot machine at any time can only be assigned to one player in one of the experimental conditions. That is, if a participant decides to remain with that slot in the next round, then that slot will be unavailable for other participants. If the participant decides to switch from a slot machine, then he or she will never be assigned again to the same slot. To continue our analogy with dating services, the above represents the case where users date only one partner at a time; the partner is not dating anyone else; and some users are more universally appealing than others.

Rewards of the slots are stochastic and their mean payoffs are known in advance by the central designer, but not to the participants. Formally, if player $i$ is assigned to slot $j$, he or she receives a payoff following a normal distribution with mean $\mathbb{E}[p_{ij}]=w_j+\epsilon_i$ (and variance set to $3$). Hence, users need repeated interactions to estimate the underlying reward.
To compare the matching algorithms, the mean payoffs of players from slot machines are the same across the two matching conditions, e.g., the first participant assigned to the Fair condition has the same mean payoff per slot machine as the first participant assigned to the Selfish condition. 

At any round, the participants may choose to stop being matched by the algorithm and take an outside payment of 6 cents for each round that is left unplayed, or may choose to continue playing in the hope of making a greater profit. We ran three studies with groups of $2n=6$ participants simultaneously and $m=13$ slots machines, that only differ in how the payoffs $p_{ij}$'s are distributed: Study~A represents a competitive market (80-20 rule) where payoffs follow a right-skewed distribution; Study~B has payoffs following a symmetric distribution (similar to a normal), and Study~C represents a non-competitive market where payments follow a left-skewed distribution. The top histograms in Figure~\ref{fig:Distribution-of-payoffs} show the distribution for each study.
The underlying distribution of possible payoffs also correspond to how payoffs would look like if the matching algorithm was randomly assigning users to slots. The lower histograms in Figure~\ref{fig:Distribution-of-payoffs} show the distribution or realized payoffs that users got in each experimental condition.
Compared to the underlying distribution, the Fair algorithm is able to transform the distribution of payoffs, increasing the mean payoff and shifting the distribution to the right in all three studies. On the other hand, the Selfish algorithm performs worse than the underlying distribution, and is therefore inferior to random matching.

\begin{figure}[t]
\begin{centering}
\subfloat[\label{fig:Study_A_SW}Study~A: Mean Utility]{\begin{centering}
\includegraphics[width=0.28\linewidth]{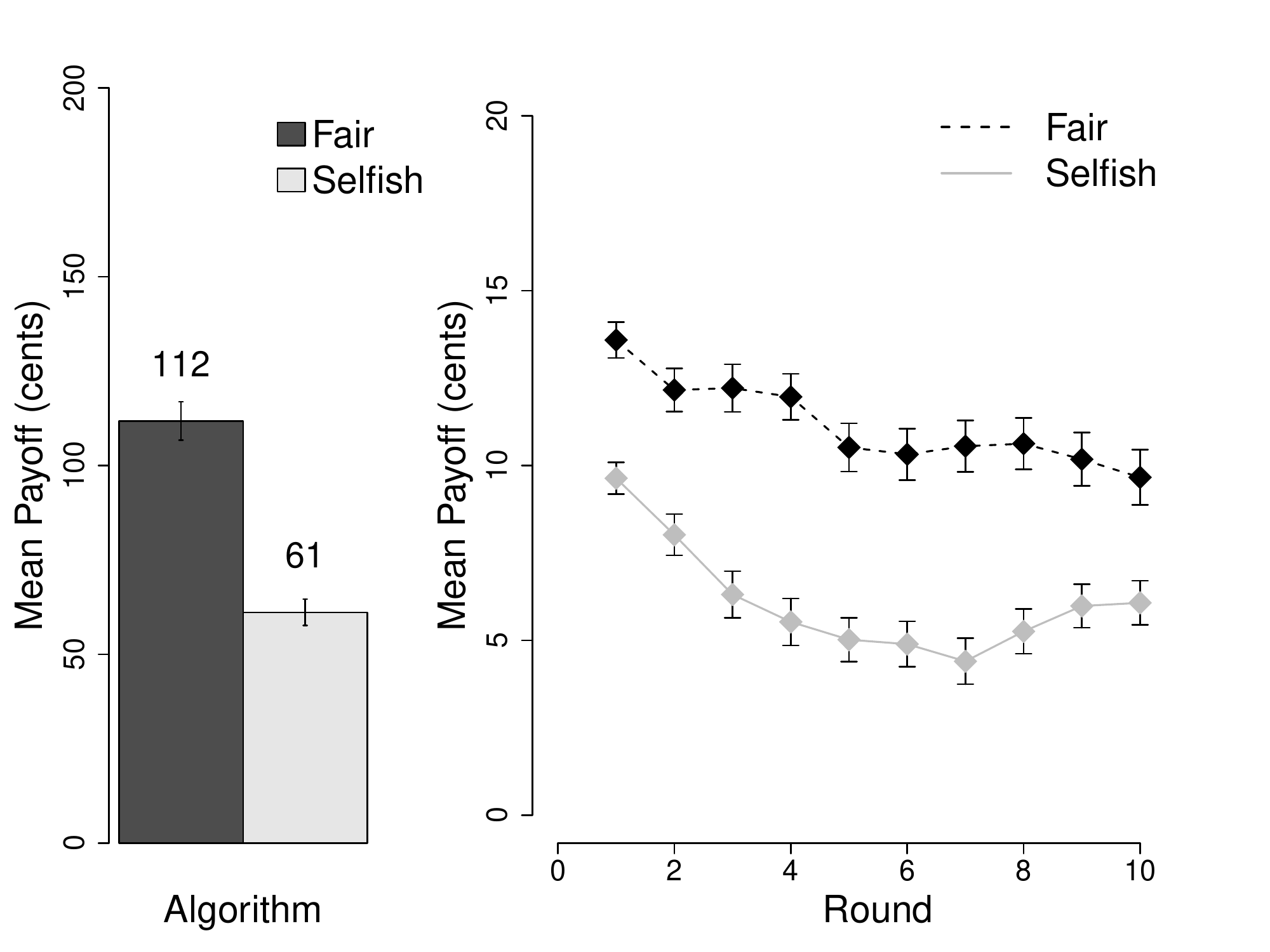}
\par\end{centering}
}\subfloat[\label{fig:Study_B_SW}Study~B: Mean Utility]{\begin{centering}
\includegraphics[width=0.28\linewidth]{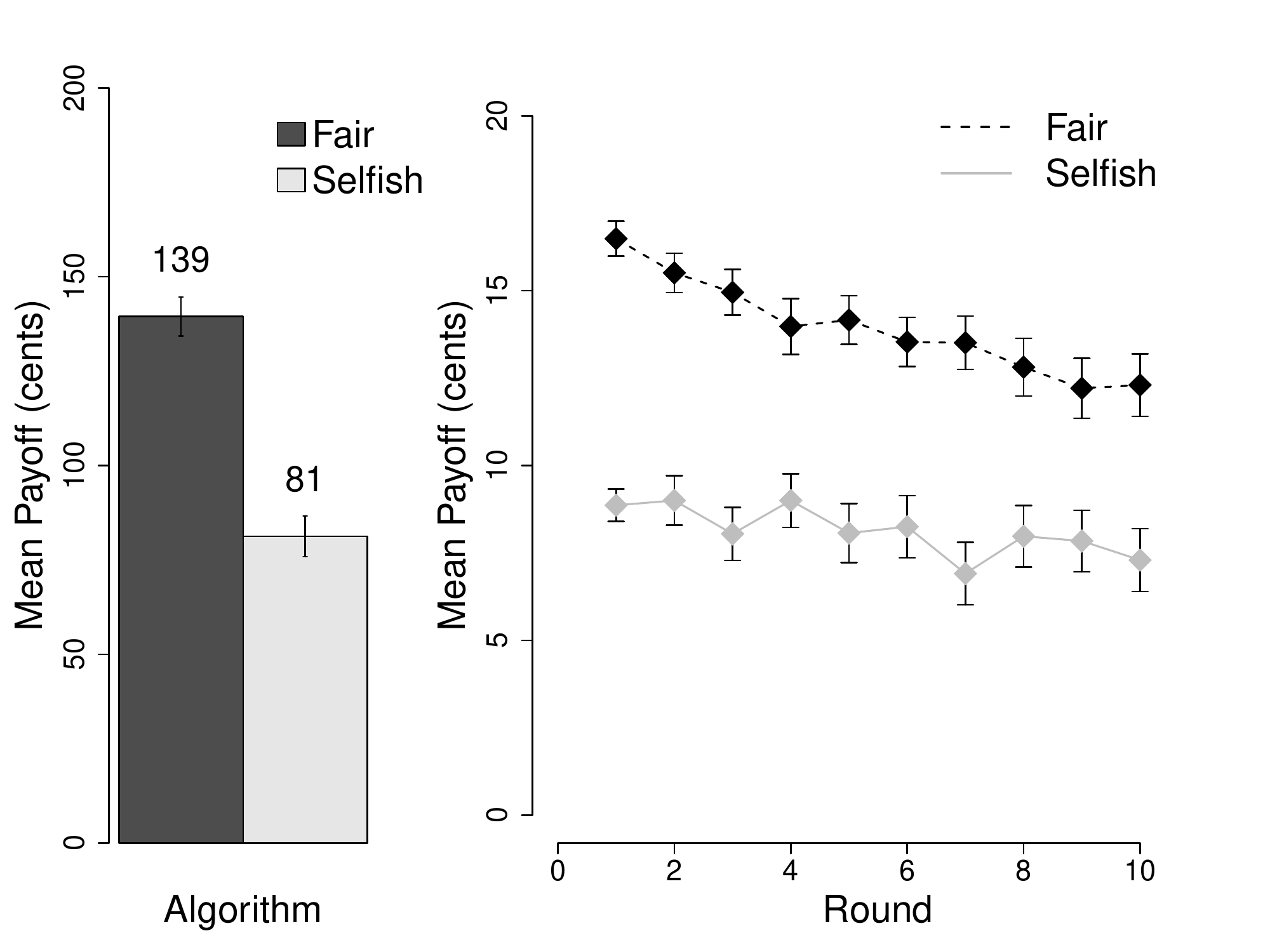}
\par\end{centering}
}\subfloat[\label{fig:Study_C_SW}Study~C: Mean Utility]{\begin{centering}
\includegraphics[width=0.28\linewidth]{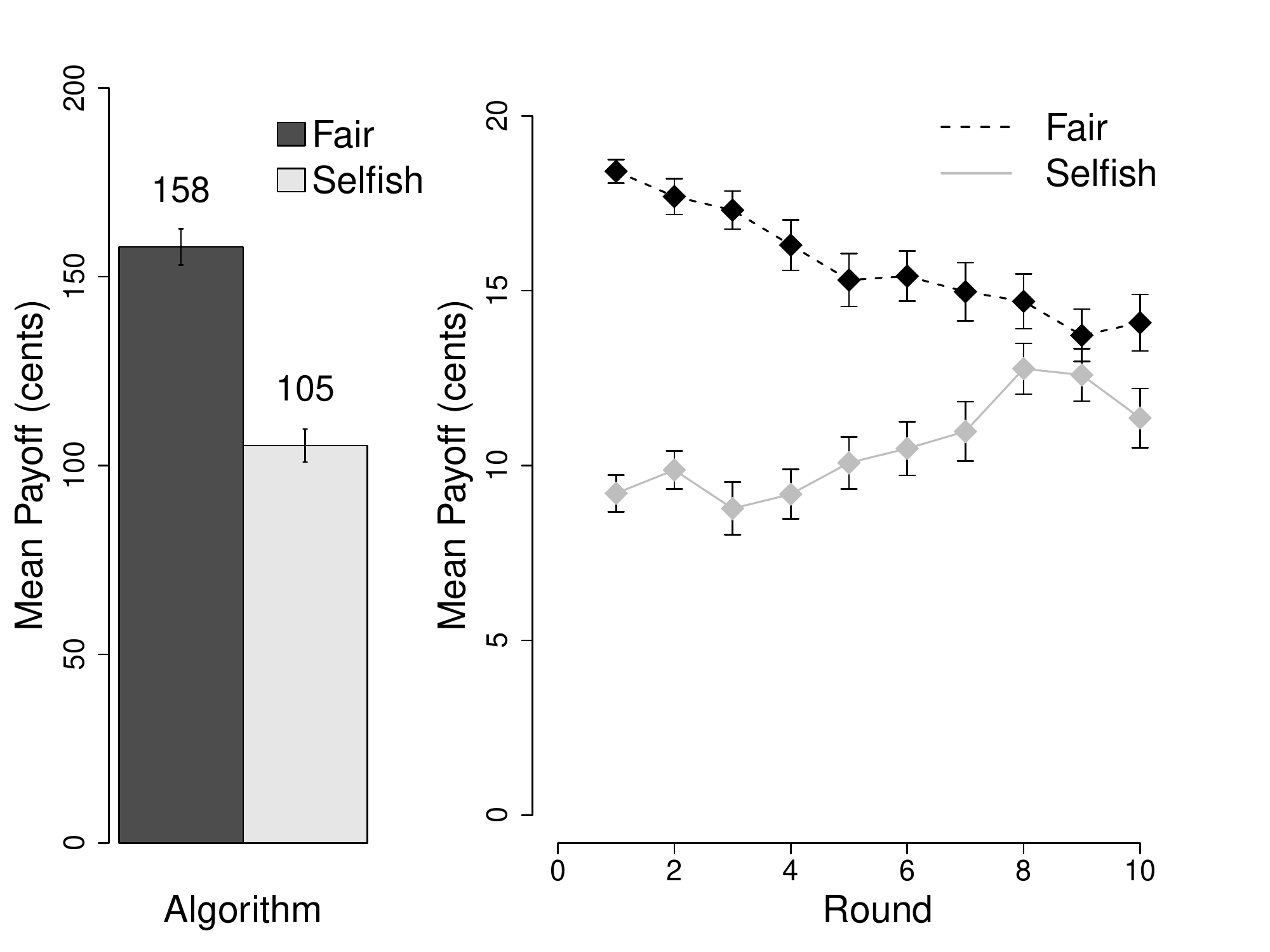}
\par\end{centering}
}
\par\end{centering}
\vspace{-0.45cm}
\begin{centering}
\subfloat[\label{fig:Study_A_eng}Study~A: Engagement Rate]{\begin{centering}
\includegraphics[width=0.28\linewidth]{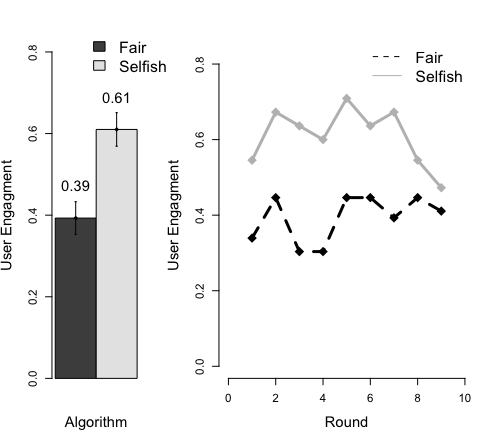}
\par\end{centering}
}\subfloat[\label{fig:Study_B_eng}Study~B: Engagement Rate]{\begin{centering}
\includegraphics[width=0.28\linewidth]{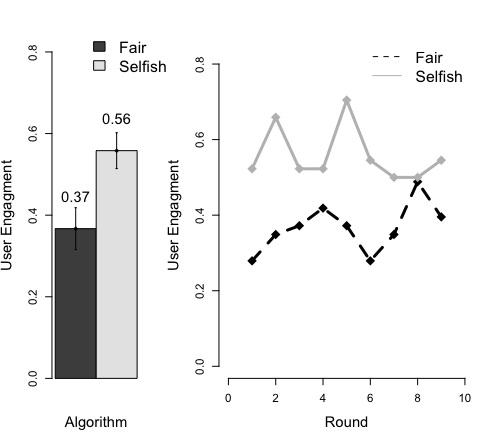}
\par\end{centering}
}\subfloat[\label{fig:Study_C_eng}Study~C: Engagement Rate]{\begin{centering}
\includegraphics[width=0.28\linewidth]{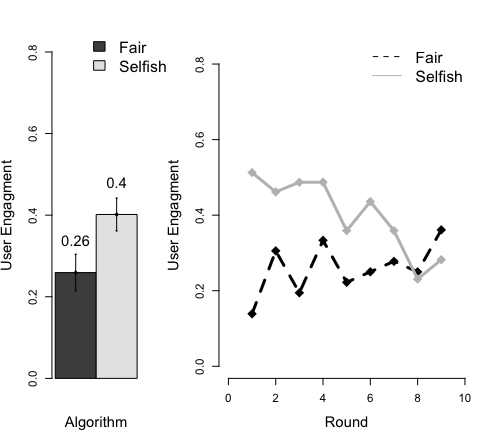}
\par\end{centering}
}
\par\end{centering}

\caption{\label{fig:SW}Figures a-c: In each sub-figure, the left plot shows the mean overall
payment received by participants; the right plot shows the mean payoffs
received per round (standard error bars
are calculated over the mean per groups). Figures d-f: In each sub-figure, the left plot shows the mean overall engagement rate by participants; the right plot shows the engagement rates per round, where engagement rate is defined as the percentage of users requesting to be re-matched.}
\end{figure}


Figure~\ref{fig:SW} depicts the mean payments and engagement rates of participants in each study.
Regarding the payments, the Fair matching algorithm achieves an overall greater utility than the Selfish algorithm. This difference is at most double the utility achieved by the Selfish algorithm, and it decreases across studies as the market becomes less competitive.
Naturally, in the Fair algorithm, payoffs decrease over time given that by definition, each assignment maximizes the social welfare. Thus, any subsequent changes of assignation can only decrease the overall utility, but will not necessarily decrease the payment of an individual participant. 
Focusing on engagement rates, defined as the percentage of users requesting to be re-matched, we find that the Selfish algorithm (whose objective function is to maximize engagement) yields at lest 50\% higher rates than the Fair matching. 
 Figure S8 shows the average PoA over pairs of instances whose initial conditions are the same. We note that for some rounds the maximum was greater than 1, meaning that the Selfish algorithm performed better than the Fair algorithm in some cases. 


Overall, our main findings are as follows:
As the markets become less competitive, the loss of social welfare generated by the Selfish matching algorithm decreases, while the gain in engagement rates across studies remains constant. At the same time, the Selfish algorithm does not induce any significant change over the rate of users taking the outside option compared to the Fair algorithm (See Figure S9).
In the extreme competitive case, Study~A, the mean utility of users under the Selfish algorithm was exactly what they could have gotten as an outside payment had they decided not to play. This is a surprising result given that the algorithm was totally agnostic about the outside payment and its value. Rather, the algorithm learns, based on the implicit feedback given by users' actions, to maximize the probability that a user clicks again. Through this feedback loop, the algorithm self regulates the system to keep payments near but above the baseline to keep engagement high at a minimum cost to the system. Thus, suggesting that users' implicit expectations of their potential utility, both in the system and out, drive the Selfish matching algorithm to perform in different ways.

\section{Related Work} 
Closely related to our question is the concern of inefficiencies that can arise due to the behavior of users in online dating platforms. 
Kanoria and Saban~\cite{kanoria2017facilitating} present a two-sided matching market, where strategic agents on both sides are able to screen potential partners, and a match results only upon approval by both sides of the market. The authors characterize the equilibria as a function of the screening costs that agents in both sides of the market incur when looking for matches. Then, they show that simple interventions such as to block one of the sides from screening, leading to a one-sided search, can improve the social welfare. There is empirical evidence of such asymmetries when screening for partners, for example, women are much more selective than men in {\em Tinder}~\cite{tyson2016first}. Empirical studies in online dating sites are also looking to estimate mate preferences, and have found that both women and men have a strong preference for similarity along many attributes~\cite{hitsch2010makes}.
Using the estimated mate preferences based on a particular online dating site, one study shows that the online dating site achieves near optimal matching compared with the optimal matching predicted by the Gale-Shapley algorithm in a two-sided market~\cite{hitsch2010matching}.

\section{Discussion and Open Problems}
Our results shed new light on how self-interest behavior by online dating sites impacts the social welfare of its users. We introduce and study a model, based on the classic matching problem on graphs, to represent self-interest behavior by dating sites.  We quantify the social efficiency of dating sites using the notion of Price-of-Anarchy to capture how much the system degrades users' utility due to selfish behavior of the agent. We establish theoretical bounds that do not depend on the number of users for the Price of Anarchy, regardless of whether the system knows beforehand all users or if they arrive sequentially (See SI Section S3). Further, we put to experimental test our model, and create a matching market with human subjects to compare the social welfare achieved by an optimal matching service against a selfish matching algorithm. 

Overall, we have optimistic results suggesting that the dating apocalypse\footnote{Tinder and the Dawn of the ``Dating Apocalypse'', Vanity Fair, 2015} is not here, yet. Indeed, a recent study suggests that online couples have lower odds of getting married than offline couples \cite{paul2014online}, however, conditioned on being married, online married couples are slightly less likely to result in marital break-ups~\cite{Cacioppo2013}.
This study is a first step towards understanding the impact of algorithms on society, and opens up several problems and challenges. Our daily lives are increasingly being affected by black-box systems, and we advocate the need for more regulation of these systems. It is our responsibility as researchers to understand, quantify and inform policy makers of the possible---intentional or otherwise---side effects of algorithms.



\bibliography{online-dating}
\bibliographystyle{plain}

\newtheorem{innercustomgeneric}{\customgenericname}
\providecommand{\customgenericname}{}
\newcommand{\newcustomtheorem}[2]{%
  \newenvironment{#1}[1]
  {%
   \renewcommand\customgenericname{#2}%
   \renewcommand\theinnercustomgeneric{##1}%
   \innercustomgeneric
  }
  {\endinnercustomgeneric}
}

\newcustomtheorem{customthm}{Theorem}
\newcustomtheorem{customlemma}{Lemma}
\newcustomtheorem{customexample}{Example}
\newcustomtheorem{customobservation}{Observation}

\setcounter{figure}{0}
 
\renewcommand{\thesection}{S\arabic{section}}   
\renewcommand{\thefigure}{S\arabic{figure}}
\renewcommand{\thetable}{S\arabic{table}}

\section*{Supplementary Information} 

\subsection*{Supplementary Proofs}

\begin{customobservation}{\ref{obs:1}}
$\pi_{i,1}(u_i)$ satisfies (A1), (A2) and (A3) (with $q_i$ replaced by $\pi_{i,1}$).
\end{customobservation}
\begin{proof}
(A1) and (A2) are immediate from \raf{eq:pi}. To see that $\pi_{i,1}(u_i)$ is concave in $u_i\in[0,1]$, note that its second derivative is given by
\begin{eqnarray*}
\pi_{i,1}''(u_i)&=&\frac{(1+q_i(u_i))q_i''(u_i)-2(q_i'(u_i))^2}{(1+q_i(u_i))^3},
\end{eqnarray*}
which is strictly negative for $u_i\in[0,1]$ by  the strict concavity assumption on $q_i(\cdot)$.\qed
\end{proof}

\vspace{20pt}

\begin{customthm}{\ref{t1}}
Let $h(\bq):=\min_{i\in\cM}\{q_i'(0)\}>0$, $H(\bq):=\max_{i\in\cM}\{q_i'(0)\}$, and $c\in(0,h(\bq))$ be the unique solution of the equation 
\begin{equation}\label{eq:c}
c=\frac{H(\bq)}{2}\cdot L(\bq,c),
\end{equation}
where $L(\bq,c)=\min_{i\in\cM}\overline{u}_i(c)$, and where for $i\in\cM$,  $\overline{u}_i(c)\in[0,1]$ denotes the unique positive solution of the equation 
\begin{equation}\label{eq:u}
\pi_{i,1}'(u_i) = \frac{q_i'(u_i)}{(1+q_i(u_i))^2}=c.
\end{equation}
Then $\PoA(\bq)\ge L(\bq,c)/2$.
\end{customthm}
\begin{proof}
First note that the solutions in \raf{eq:c} and \raf{eq:u} are well-defined, since the strict concavity of $q_i(u_i)$ in $[0,1]$ implies that the function $\pi_{i,1}'(u_i)=\frac{q_i'(u_i)}{(1+q_i(u_i))^2}$ is strictly decreasing in $[0,1]$, which in turn implies that the function $L(\bq,c)$ is also strictly decreasing in $c\in[0,h(\bq)]$. It follows that the right-hand side of \raf{eq:c} is strictly decreasing in $c\in[0,h(\bq)]$, and assumes a positve value when $c=0$, and a value of $0$ when $c=h(q)$, implying that the root $c$ in \raf{eq:c} exists.  
  
Fix any optimal solutions $(\bu^*,\bx^*)$ and $(\widehat\bu,\widehat\bx)$ for problems \raf{o-f} and \raf{o-s}, respectively.
Then our to goal is to lower-bound \begin{equation}\frac{\sum_{i\in\cM}\widehat u_i}{\sum_{i\in\cM}u_i^*}.\label{e--}\end{equation}
Fix $c>0$ as given by \raf{eq:c}, and let $\cS(c):=\{i\in\cM:~\pi_{i,1}'(\widehat u_i)>c\}$.
%
If one tries to bound the ratio directly by using the simple (local) lower bound $\min_i\frac{\widehat u_i}{u_i^*}$, one runs into the difficulty that some of the $\widehat u_i$'s could be very close to zero (even though not exactly zero as the function $s(\bu;\bq)$ is strictly concave).  To overcome this, we use the local bound only for users not in $\cS(c)$, while for users in $\cS(c)$ we use a global bound derived from the KKT optimality conditions. This is formalized in the following  two claims.   

\begin{claim}\label{cla1}
Suppose that $i\in \cM\setminus\cS(c)$, then $\widehat u_i\ge \frac{L(\bq,c)}{2}(u_i^*+\widehat u_i)$.
\end{claim}
\begin{proof}
Immediate from the definitions since $\pi'_{i,1}(\widehat u_i)\le c$ implies, by the concavity of $q_i$, that $\widehat u_i \ge \overline{u}_i(c)\ge L(\bq,c)\ge L(\bq,c)(u_i^*+\widehat u_i)/2$. 
\qed
\end{proof}

\begin{claim}\label{cla2}
$\sum_{i\in\cM}u_i^*\le \frac{1}{c}\sum_{i\in \cS(c),~j\in\cW}\pi_{i,1}'(\widehat u_i)w_{ij}\widehat x_{ij}+\sum_{i\in \cM\setminus \cS(c)}(u_i^*+\widehat u_i).$
\end{claim}
\begin{proof}
Consider problem \raf{o-s}. By the KKT (necessary) conditions for optimality (which are also sufficient since the functions $\pi_{i,1}(\cdot)$ are concave by Observation~\raf{ob1}), there exist $\widehat\beta_i\ge0$, $\widehat\sigma_j \ge 0$, and $\widehat\mu_{ij}\ge 0$, for $i\in \cM$ and $j\in\cW$, such that \raf{e1}-\raf{e4} hold for $(\bu,\bx)=(\widehat\bu,\widehat\bx)$, and
\begin{align}
-\pi_{i,1}'(\widehat u_i)w_{ij}+\widehat\beta_i+\widehat\sigma_j&=\widehat\mu_{ij},~~\text{ for all }i\in \cM\text{ and }j\in \cW,\label{kkt-e1}\\
\qquad \widehat\mu_{ij}\widehat x_{ij}&= 0,~~\text{ for all }i\in\cM\text{ and }j\in\cW,\label{kkt-e2}\\
\qquad\widehat\beta_i(\sum_{j\in\cW}\widehat x_{ij}-1)&=0,~~\text{ for all }i\in\cM,\label{kkt-e3}\\
\qquad\widehat\sigma_j(\sum_{i\in\cM}\widehat x_{ij}-1)&=0,~~\text{ for all }j\in\cW.\label{kkt-e4}
\end{align}
Note that \raf{e2}, \raf{e3}, and \raf{kkt-e1}-\raf{kkt-e4} imply that
\begin{align}
\sum_{i\in\cM,~j\in\cW}\pi_{i,1}'(\widehat u_i)w_{ij}\widehat x_{ij}&=\sum_{i\in\cM}\widehat\beta_i+\sum_{j\in\cW}\widehat\sigma_j,\label{kkt-e1-}\\
\qquad \widehat\beta_i+\widehat\sigma_j &\ge \pi_{i,1}'(\widehat u_i)w_{ij},~~\text{ for all }i\in \cM\text{ and }j\in \cW.\label{kkt-e2-}
\end{align}

Let us next define a new set of weights $\widetilde\bw$, obtained from $\bw$ by setting $\widetilde w_{ij}=0$ for all $i\in\cM\setminus\cS(c)$ and $j\in\cW$. Let $(\widetilde\bu,\widetilde\bx)$ be an optimal solution to problem $O_f(\widetilde\bw)$.
Clearly, by the feasibility of $u_i^*$ for \raf{e2} and \raf{e3},
\begin{equation}\label{eq1}
\sum_{i\in\cM}u_i^*\le \sum_{i\in\cS(c)}\widetilde u_i+\sum_{i\in\cM\setminus\cS(c)}u_i^*.
\end{equation}
   
Let us write the dual for LP $O_f(\widetilde\bw)$ (in the variables $\beta:=(\beta_i:~i\in\cS(c))$ and $\sigma:=(\sigma_1,\ldots,\sigma_n)$):
\begin{align}
z^*_f(\widetilde\bw):=&\min_{\beta,\sigma}\sum_{i\in \cS(c)}\beta_i+\sum_{j\in \cW}\sigma_j \label{dual}\tag{$D_f(\widetilde\bw)$}\\
\text{s.t.}& \qquad \beta_i+\sigma_j \ge w_{ij},~~\text{ for all }i\in \cS(c)\text{ and }j\in \cW,\label{d-e1}\\
&\qquad \beta_i\ge 0,~\sigma_{j}\ge 0,~~\text{ for all } i\in\cS(c) \text{ and } j\in \cW.\label{d-e2}
\end{align}
From \raf{kkt-e2-}, we have for $i\in\cS(c)$ and $j\in \cW$: $\widehat\beta_i+\widehat \sigma_j\ge c \cdot w_{ij}$. Thus, the pair of vectors $(\beta',\sigma'):=((\frac{\widehat\beta_i}{c}:~i\in\cS(c)),(\frac{\widehat\sigma_j}{c}:~j\in\cW))$ is feasible for the dual LP $O_f(\widetilde\bw)$. It follows that 
\begin{eqnarray}
\sum_{i\in\cS(c)}\widetilde u_i&=&z^*_f(\widetilde\bw)~~~~~~~~~~\text{(by LP duality)}\nonumber\\
&\le& \frac{1}{c}\left(\sum_{i\in \cS(c)}\widehat\beta_i+\sum_{j\in \cW}\widehat\sigma_j\right)~~~\text{ (by the feasibility of $(\beta',\sigma')$ for ($D_f(\widetilde\bw)$))}\nonumber\\
&\le&\frac{1}{c}\left(\sum_{i\in\cS(c),~j\in\cW}\pi_{i,1}'(\widehat u_i)w_{ij}\widehat x_{ij}+\sum_{i\in\cM\setminus\cS(c),~j\in\cW}\pi_{i,1}'(\widehat u_i)w_{ij}\widehat x_{ij}\right)~~~\text{ (by \raf{kkt-e1-})}\nonumber\\
&\le&\frac{1}{c}\left(\sum_{i\in\cS(c),~j\in\cW}\pi_{i,1}'(\widehat u_i)w_{ij}\widehat x_{ij}+c\sum_{i\in\cM\setminus\cS(c)}\sum_{j\in\cW}w_{ij}\widehat x_{ij}\right)~~~\text{(by definition of $\cS(c)$)}\nonumber\\
&=&\frac{1}{c}\left(\sum_{i\in\cS(c),~j\in\cW}\pi_{i,1}'(\widehat u_i)w_{ij}\widehat x_{ij}+c\sum_{i\in\cM\setminus\cS(c)}\widehat u_i\right)~~~\text{(by feasibility of $\widehat\bx$ for \raf{o-f}).}\nonumber\\
\label{eq2}
\end{eqnarray}  
Using \raf{eq2} in \raf{eq1}, we arrive at the claim.\qed
\end{proof}

Finally, using the above two claims, we get
\begin{eqnarray}
\frac{\sum_{i\in\cM}\widehat u_i}{\sum_{i\in\cM}u_i^*}&\ge&\frac{\sum_{i\in\cS(c),~j\in\cW}w_{ij}\widehat x_{ij}+\sum_{i\in\cM\setminus\cS(c)}\widehat u_i}{\frac{1}{c}\sum_{i\in\cS(c),~j\in\cW}\pi_{i,1}'(\widehat u_i)w_{ij}\widehat x_{ij}+\sum_{i\in\cM\setminus\cS(c)}(u_i^*+\widehat u_i)}\nonumber\\
&\ge&\min\left\{c\cdot\min_{i\in\cS(c)}\frac{1}{\pi_{i,1}'(\widehat u_i)},\min_{i\in\cM\setminus\cS(c)}\frac{\widehat u_i}{u_i^*+\widehat u_i}\right\}\nonumber\\
&\ge&\min\left\{c\cdot\frac{1}{\max_{i\in\cS(c)}\pi_{i,1}'(0)},\frac{L(\bq,c)}{2}\right\}\label{eq-}\\
&\ge&\min\left\{\frac{c}{H(\bq)},\frac{L(\bq,c)}{2}\right\}=\frac{L(\bq,c)}{2},\label{eq--}
\end{eqnarray}
where \raf{eq-} follows from the concavity of $q_i$  and Claim~\ref{cla1}; \raf{eq--} follows from \raf{eq:c}. The theorem follows.\qed
\end{proof}

\newpage
\begin{customlemma}{\ref{lem:1}}
$\argmax_{u_i}\pi_{i,1}(u_i)=1$,
as $\epsilon \rightarrow 0$.   
\end{customlemma}
\begin{proof}
To derive the above result, we show that the global optimum $u^*_i$ of $\pi_{i,1}(u_i)$ is the unique solution satisfying,
$$
\pi_{i,1}'(u^*_i) = \frac{q_i'(u^*_i)+ \frac{q_i^2(u^*_i)}{\epsilon}}{\left(1+q_i(u^*_i)+\frac{q_i(u^*_i)}{\epsilon}(1-u^*_i)\right)^2}=0,
$$
or equivalently, 
\begin{equation}\label{eq:epsilon}
\epsilon = -\frac{q_i^2(u^*_i)}{q_i'(u^*_i)}.
\end{equation}
Given that $\epsilon>0$, Equation \raf{eq:epsilon} holds only if $q_i'(u^*_i)<0$. This in turn means that $u^*_i \in [u_i',1]$, where $u_i'$ is the solution to $q_i'(u_i')=0$. It follows from the concavity of $q$ that $-\frac{q_i^2(u_i)}{q_i'(u_i)}$ is decreasing in the interval $u_i\in [u_i',1]$. Therefore, $u^*_i$ must be unique. Furthermore, as $\epsilon \rightarrow 0$, we have that $q_i^2(u_i^*) \rightarrow 0$,  or equivalently,  $u_i^* \rightarrow 1$.
\qed
\end{proof}

\vspace{20pt}
\begin{figure}[t]
\begin{centering}
\includegraphics[width=0.4\columnwidth]{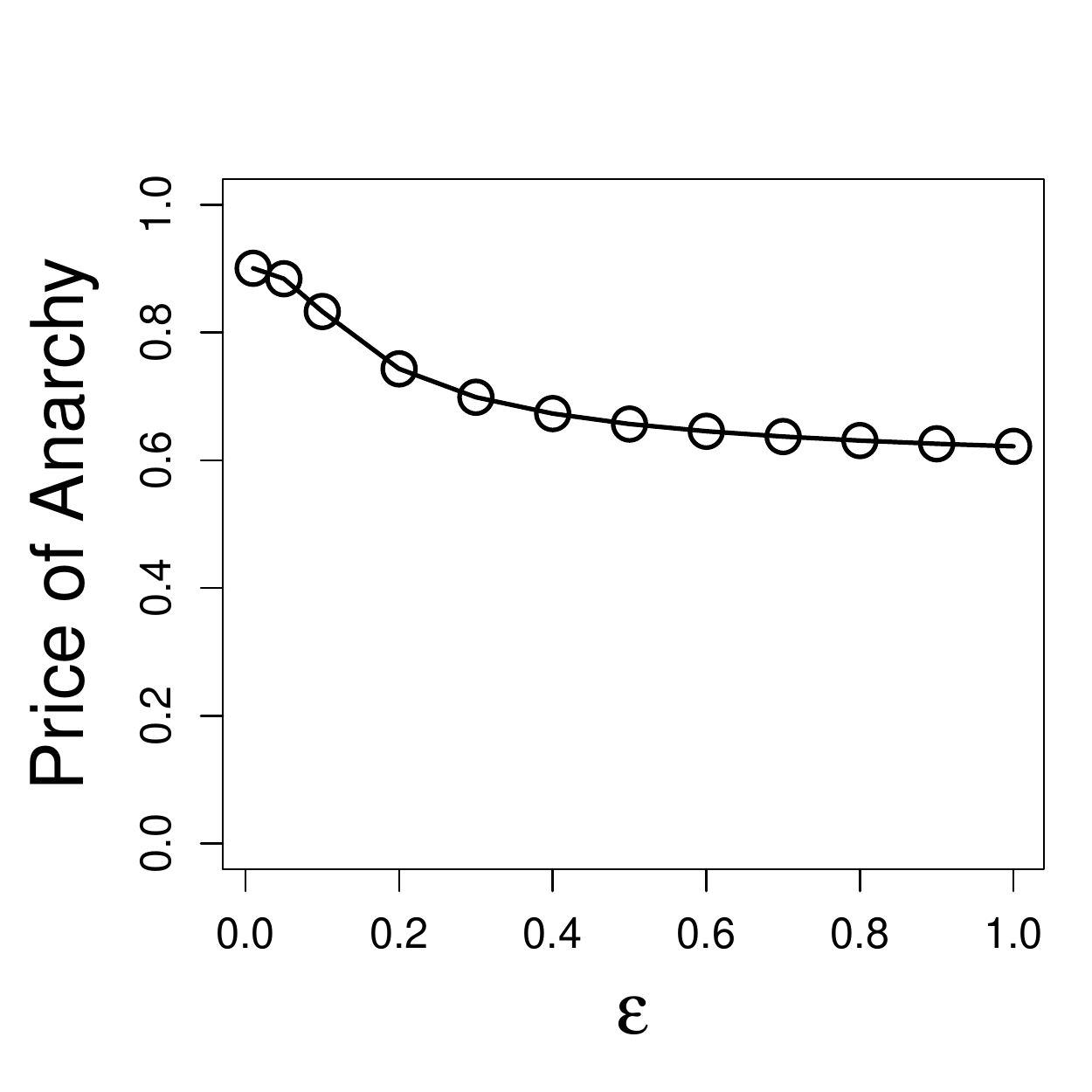}
\par\end{centering}
\caption{\label{fig:example2} The figure shows the empirical Price of Anarchy for probability function $q_i(u):=u(1-u)$, for different $\epsilon$'s. The empirical Price of Anarchy is obtained as the minimum ratio of utility between the selfish and optimal matching, across several instances randomizing the preference weights.  }
\end{figure}
\begin{customexample}{S1}\label{ex:2}
 We validate our findings with a numerical example. Figure~\ref{fig:example2} depicts the convergence predicted for the probability function $q_i(u):=u(1-u)$ using empirical results. The empirical Price of Anarchy is obtained as the minimum ratio of utility between the selfish and optimal matching, across several instances randomizing the preference weights. That is, for each instance we generate a set of random weights $\bw$ sampled from a Beta distribution (with parameters $(2,2)$). Then, we solve the optimal matching $O_f(\bw)$, the selfish matching $O_s(\bw;\bq)$ and obtain the ratio of their utilities. 
\end{customexample}


\newpage

\section*{Supplementary Text - Bounds on $\PoA$ in the Online Model}

As in the previous sections, we are interested in quantifying the efficiency of assignments for the online selfish matching problem \raf{o-s} with respect to assignments for the fair matching problem \raf{o-f}. Inspired by the notion of competitive ratio, we extend the definition of price of anarchy for the online model as follows.
\begin{definition}[Price of Anarchy for the online model]\label{PoA2}
Let $\cA$ be an online matching algorithm for solving the selfish matching problem \raf{o-s} and $s(\bu_s^{\cA};\bq)$ be the objective value of the online strategy.
Let $u^*_f\in\argmax\{f(\bu):\bu\text{ satisfies \raf{e1}-\raf{e4}}\}$.
Define the {\it price of anarchy} $\PoA^\cA=\PoA^\cA(\bq)\in[0,1]$, (from the point of view of the users) with respect to algorithm $\cA$, as:
$$
\PoA^\cA(\bq):=\min_{\bw\ge \bzero}\frac{f(\bu^{\cA}_s)}{f(\bu^*_f)}.
$$
\end{definition}

Next, we analyze the $\PoA$ of the greedy policy for the selfish (probabilistic) matching problem. The greedy algorithm, upon arrival of element $i\in \cM$, decides the (partial) distribution $\bx_i:=(x_{ij}:~i\in\cM,~j\in\cW,~\sum_{j\in\cW}x_{ij}\le 1)$ that maximizes $\pi_{i,1}(u_i)$, and assigns an element $j\in \cW$ to $i$ with probability $x_{ij}$ (with probability $1-\sum_{j\in\cW}x_{ij}$ user $i$ is not assigned); the agent may also be required to ensure that $\sum_{i\in\cM}x_{ij}\le 1$ for all $j\in\cW$.

\begin{theorem}\label{t2}
Let $\cA$ be the greedy policy for the selfish (probabilistic) matching problem. Then, 
$$ \PoA^\cA(\bq)\ge L(\bq,c)/2, $$
where $L(\bq,c)$ is defined as in Theorem \ref{t1}.
\end{theorem}
\begin{proof}
Let $i\in \cM$ be the $i$th requested element in the online algorithm. Upon arrival of element $i\in \cM$, the greedy algorithm solves:
\begin{align}
&\max_{\bx_i}\pi_{i,1}(u_i)  \\
\text{s.t.}& \qquad u_i:=\sum_j w_{ij}x_{ij},\label{eo1}\\
&\qquad\sum_{j\in\cW}x_{ij}\le 1,\label{eo2}\\
&\qquad\sum_{k\leq i}x_{kj}\le 1,~~\text{ for all }j\in \cW,\label{eo3}\\
&\qquad x_{ij}\ge 0,~~\text{ for all } j\in \cW.\label{eo4}
\end{align}
Notice that the only difference of this problem with respect to the offline model comes from the set of constraints \raf{eo3}, which specify the availability of element $j \in \cW$ at time $i$.  
By the KKT (necessary) conditions for optimality, 
there exists $\widehat\beta_i\ge0$, $\widehat\sigma_{i,j} \ge 0$, and $\widehat\mu_{ij}\ge 0$, for $i\in \cM$ and $j\in\cW$, such that \raf{eo1}-\raf{eo4} hold for $(\bu,\bx)=(\widehat\bu,\widehat\bx)$, and
\ \\
\begin{align}
-\pi_{i,1}'(\widehat u_i)w_{ij}+\widehat\beta_i+\widehat\sigma_{i,j}&=\widehat\mu_{ij},~~\text{ for all }i\in \cM\text{ and }j\in \cW,\label{kkt-eo1}\\
\qquad \widehat\mu_{ij}\widehat x_{ij}&= 0,~~\text{ for all }i\in\cM\text{ and }j\in\cW,\label{kkt-eo2}\\
\qquad\widehat\beta_i(\sum_{j\in\cW}\widehat x_{ij}-1)&=0,~~\text{ for all }i\in\cM,\label{kkt-eo3}\\
\qquad\widehat\sigma_{i,j}(\sum_{k\leq i}\widehat x_{kj}-1)&=0,~~\text{ for all }i\in \cM\text{ and } j\in\cW.\label{kkt-eo4}
\end{align}
For any  $j\in \cW$,  taking the sum of constraints in \raf{kkt-eo4} over $i \in \cM$ gives,
\ \\
\[
\sum_{i \in \cM}\widehat \sigma_{i,j} x_{ij} =   \sum_{i \in \cM}\widehat \sigma_{i,j} (1-\sum_{k< i}x_{kj}) \geq \sigma_{1,j},
\]
\ \\
which implies that 
\ \\
$$\sum_{i\in\cM,~j\in\cW}\widehat \sigma_{i,j} x_{ij} \geq \sum_{j \in \cW}  \sigma_{1,j}.$$
\ \\
From this formula,  together with  \raf{eo2},\raf{eo3}, \raf{kkt-eo1}-\raf{kkt-eo4}, it follows that
\ \\
\begin{align}
\sum_{i\in\cM,~j\in\cW}\pi_{i,1}'(\widehat u_i)w_{ij}\widehat x_{ij}&\geq \sum_{i\in\cM}\widehat\beta_i+\sum_{j\in\cW}\widehat\sigma_{1,j},\label{kkt-eo1-}\\
\qquad \widehat\beta_i+\widehat\sigma_{1,j} &\ge \pi_{i,1}'(\widehat u_i)w_{ij},~~\text{ for all }i\in \cM\text{ and }j\in \cW.\label{kkt-eo2-}
\end{align}
\ \\
From \raf{kkt-eo2-}, we have for $c>0, i\in\cS(c)$ and $j\in \cW$: $\widehat\beta_i+\widehat \sigma_{1,j}\ge c \cdot w_{ij}$. Thus, the pair of vectors $(\beta',\sigma'):=((\frac{\widehat\beta_i}{c}:~i\in\cS(c)),(\frac{\widehat\sigma_{1,j}}{c}:~j\in\cW))$ is feasible for the dual LP $O_f(\widetilde\bw)$. As such, the result follows from \raf{eq2}, \raf{eq1} and \raf{eq--}. \qed
\end{proof}

\newpage
\section*{Supplementary Text - Experimental Study}

We ran three studies with groups of $2n=6$ participants playing simultaneously; these were recruited using Amazon Mechanical Turk. Each study consisted of $m=13$ slots machines, and the studies only differ in how the payoffs $w_j$'s were distributed: Study~A represents a competitive market where payments follow a right-skewed Beta distribution $\beta(1,2)$; Study~B has payments following a symmetric Beta distribution $\beta(2,2)$, and Study~C represents a non-competitive market where payments follow a left-skewed Beta distribution $\beta(2,1)$. Table~\ref{table1} shows a summary of the users in each of the three studies.

\begin{table}[h]
\begin{center}
\begin{tabular}{|l|l|l|l|l|}
\hline 
\multirow{2}{*}{} & \multirow{2}{*}{Participants} & \multirow{2}{2cm}{Players per group} & \multirow{2}{*}{Outside payoff} & \multirow{2}{2cm}{Distribution of payoffs}\tabularnewline
 &  &  &  & \tabularnewline
\hline 
\hline 
Study~A & 113 & 3 & 6 cents & $\beta(1,2)$\tabularnewline
\hline 
Study~B & 87 & 3 & 6 cents & $\beta(2,2)$\tabularnewline
\hline 
Study~C & 87 & 3 & 6 cents & $\beta(2,1)$\tabularnewline
\hline 
\end{tabular}
\par\end{center}
\caption{\label{table1}Summary of participants for each experiment.}
\end{table}

\begin{center}
\begin{figure}[hbtp]
\begin{centering}
\subfloat[Screenshot of the interface shown at each round to participants. Users are shown their last payoff, their accumulated payment so far, and the three options for the next round. ]{
\begin{centering}
\includegraphics[width=0.9\textwidth]{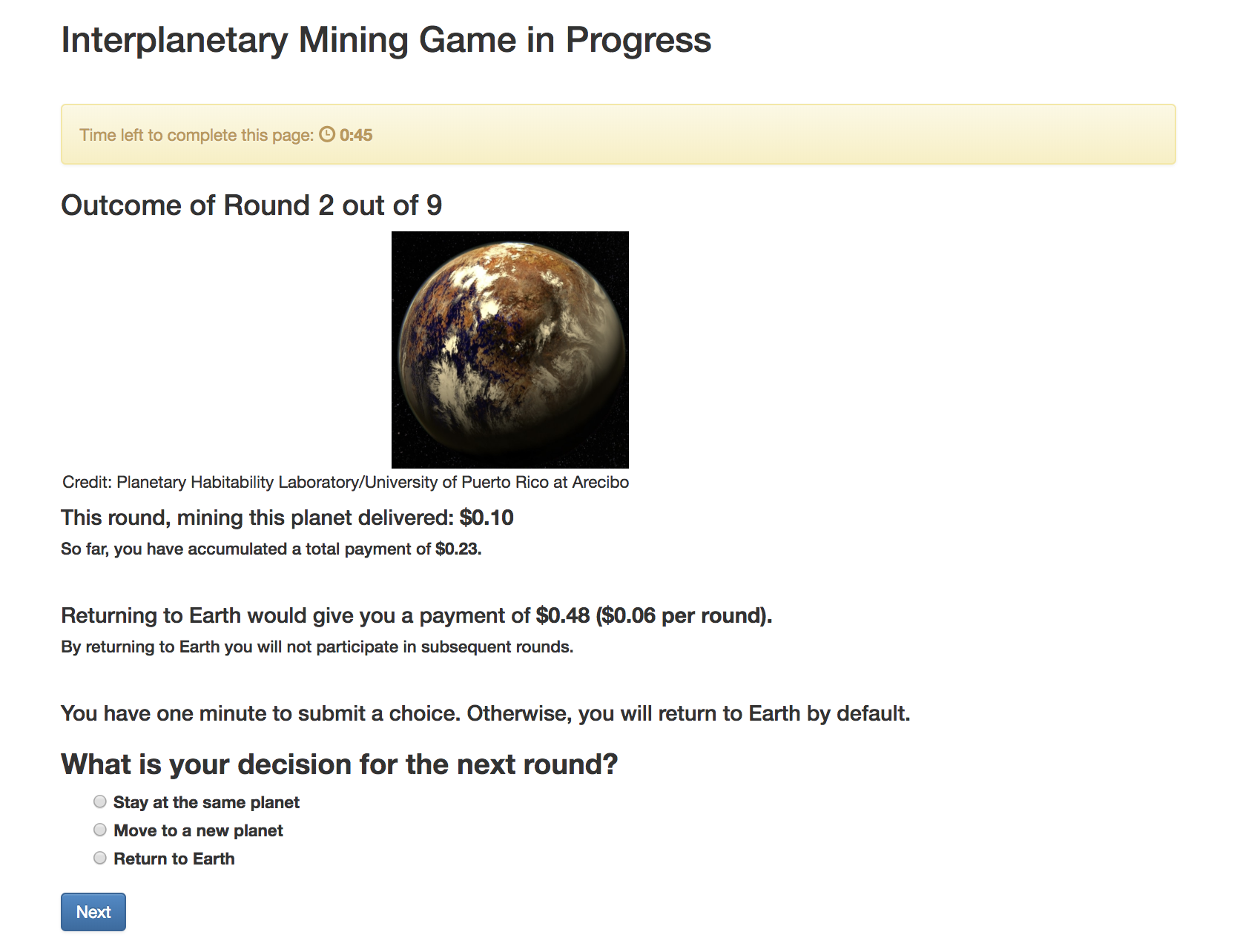}
\par\end{centering}
}\par
\vspace{20pt}
\subfloat[Screenshot of the waiting page after participants take a decision at each round. This is displayed until all 6 users in the game have taken a decision for the current round.]{
\begin{centering}
\includegraphics[width=0.9\textwidth]{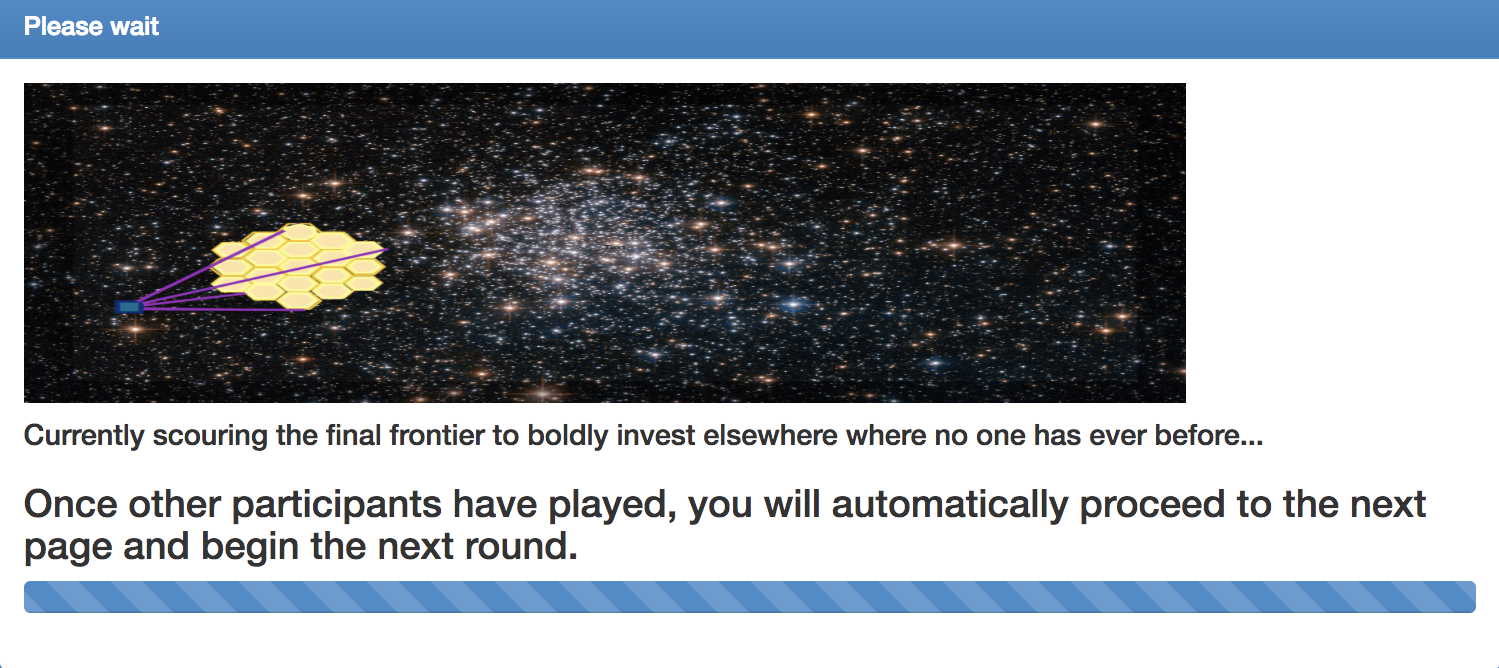}
\par\end{centering}
}
\par\end{centering}
\caption{\label{fig:screenshot} Screenshots of the experiments shown to participants.}
\end{figure}
\par\end{center}

\newpage
The algorithm used to match players takes two forms. The Fair matching algorithm assigns agents to slot machines in order to maximize the social utility of all agents.  The Selfish matching algorithm assigns agents to slot machines in order to maximize the number of agents requesting re-assignations to the system. 
Therefore, the Fair matching algorithm is given by solving problem $O_f(\bw)$. However, for the selfish matching, in order to solve $O_s(\bw;\bq)$ we need to learn the probability $q_i(u)$ that a user having received utility $u$, will come back to the system on round $i$. Given that individual data is sparse, we assume that the users are homogeneous, and start with a prior probability $q_1(u)=u(1-u)$ satisfying assumptions (A1) to (A3). Then, we update the probability based on the past choices of participants as follows:
\ \\
\[
q_{i+1}(u) = \alpha q_i(u) + (1-\alpha)f_i(u),
\]
\ \\
where $f_i(u)$ is the fraction of participants who, having received a payoff $u$ at round $i$, requested to be re-assigned. In other words, $f_i(u)$ is the empirical probability of switching, defined as the ratio of participants who requested a re-match from the system, having received a payoff of $u$ at round $i$.
Figures \ref{fig:prob_switch_A}, \ref{fig:prob_switch_B} and \ref{fig:prob_switch_C} depict the distribution $f_i(u)$ for the different rounds and algorithms in studies A, B and C, respectively. As can be seen, the switching rates of participants who received payoffs below 6 cents are surprisingly high, given that they could choose the risk-free outside option of 6 cents. Additionally, we observe a decreasing trend of switching rates for high payments (16-20 cents) as the number of rounds advances.
Taken together, these results suggest that participants are risk takers, but their propensity to take risk decreases in the final rounds of the game.


Figure~\ref{fig:PoA_ex} shows the average PoA over pairs of instances whose initial conditions are the same. We note that for some rounds the maximum was greater than 1, meaning that the Selfish algorithm performed better than the Fair algorithm in
some cases. For example, in Study~A, for some instances the social welfare under the Selfish algorithm was up to 1.5 times greater than the social welfare of the Fair algorithm.

In the main paper, we show that as markets become less competitive, the loss in social welfare generated by the Selfish matching algorithm decreases, however, the gain of engagement rates across studies remains constant. At the same time, the Selfish algorithm does not induce any significant change over the rate of users taking the outside option compared with the Fair algorithm; see Figure~\ref{fig:SI_drop}.
Drop rates, i.e., the percentage of participants that took the outside option, is slightly higher for the Fair algorithm, however, the difference is statically insignificant. 
Overall, the Selfish algorithm increases engagement rates, without increasing the drop rates.

\begin{figure}[h]
\begin{centering}
\subfloat[Fair Algorithm ]{
\includegraphics[width=0.85\textwidth]{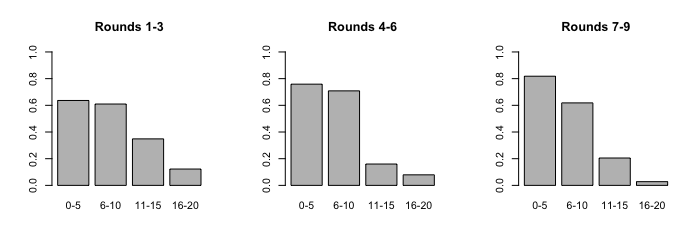}
}\par
\subfloat[Selfish Algorithm]{
\includegraphics[width=0.85\textwidth]{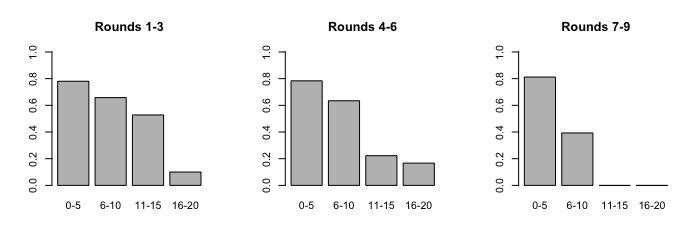}
}
\par\end{centering}
\caption{\label{fig:prob_switch_A} Probability of switching for Study~A. Distribution of the fraction of participants having received a payoff $u$ at round $i$ that requested to be re-assigned.}
\end{figure}

\begin{figure}[h]
\begin{centering}
\subfloat[Fair Algorithm ]{
\includegraphics[width=0.85\textwidth]{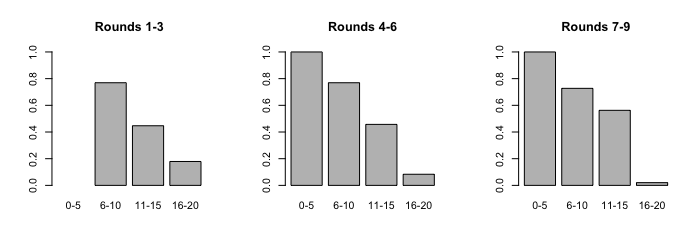}
}\par
\subfloat[Selfish Algorithm]{
\includegraphics[width=0.85\textwidth]{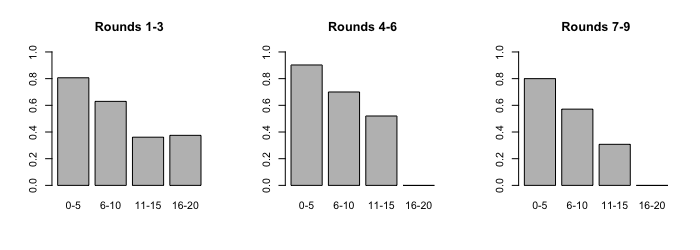}
}
\par\end{centering}
\caption{\label{fig:prob_switch_B} Probability of switching for Study~B. Distribution of the fraction of participants having received a payoff $u$ at round $i$ that requested to be re-assigned. Blank columns correspond to no user receiving a payoff in that corresponding range.   }
\end{figure}

\begin{figure}[h]
\begin{centering}
\subfloat[Fair Algorithm ]{
\includegraphics[width=0.85\textwidth]{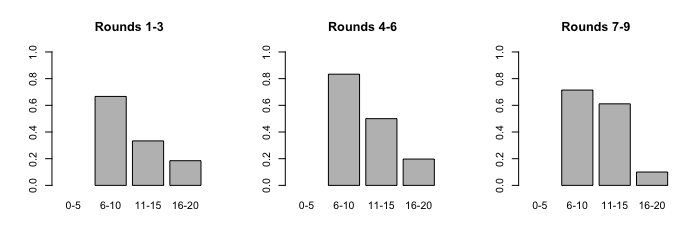}
}\par
\subfloat[Selfish Algorithm]{
\includegraphics[width=0.85\textwidth]{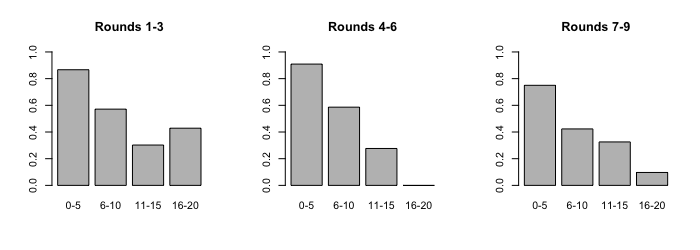}
}
\par\end{centering}
\caption{\label{fig:prob_switch_C}  Probability of switching for Study~C. Distribution of the fraction of participants having received a payoff $u$ at round $i$ that requested to be re-assigned. Blank columns correspond to no user receiving a payoff in that corresponding range. }
\end{figure}

\begin{figure}[h]
\begin{centering}
\includegraphics[scale=0.4]{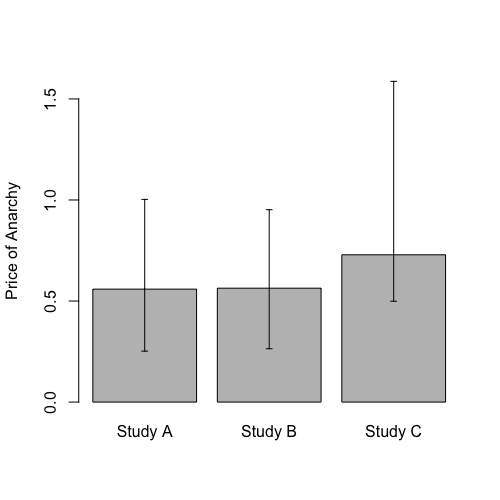}
\par\end{centering}
\caption{\label{fig:PoA_ex} Average Price of Anarchy for the three different studies. The plot shows the mean Price of Anarchy between the fair and selfish algorithm, controlling for initial conditions. Error bars depict the maximum and minimum levels.}
\end{figure}

\begin{figure}[h]
\begin{centering}
\subfloat[\label{fig:drop_A}Study~A]{\begin{centering}
\includegraphics[scale=0.35]{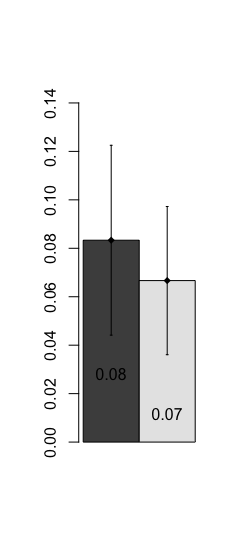}
\par\end{centering}
}\subfloat[\label{fig:drop_B}Study~B]{\begin{centering}
\includegraphics[scale=0.35]{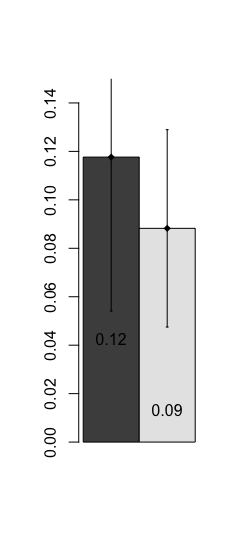}
\par\end{centering}
}\subfloat[\label{fig:drop_B}Study~C]{\begin{centering}
\includegraphics[scale=0.35]{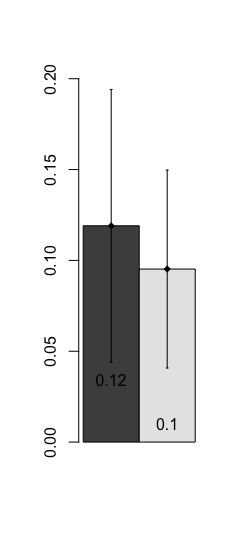}
\par\end{centering}
}
\par\end{centering}
\caption{\label{fig:SI_drop} Drop Rate. Drop rate is defined as the percentage of users taking the outside payment and thus, leaving the system. }
\end{figure}

\end{document}